\documentclass[runningheads,a4paper]{llncs}

\usepackage{makecell}
\usepackage{multirow}
\usepackage{float}
\usepackage{amssymb}
\usepackage{amsmath, bm}
\usepackage{bm}
\usepackage{algorithm}
\usepackage{algorithmic}
\usepackage{enumerate}
\usepackage{cite}
\usepackage{color}
\usepackage{adjustbox}
\usepackage{blindtext}

\newtheorem{myDef}{Definition}
\usepackage{threeparttable}
\usepackage{booktabs}

\usepackage[hidelinks]{hyperref} 
\hypersetup{
colorlinks=true,
linkcolor=blue,
citecolor=blue
}

\usepackage{graphicx}
\usepackage{tikz,xcolor,hyperref}
\definecolor{lime}{HTML}{A6CE39}
\DeclareRobustCommand{\orcidicon}{%
    \begin{tikzpicture}
    \draw[lime, fill=lime] (0,0) 
    circle [radius=0.16] 
    node[white] {{\fontfamily{qag}\selectfont \tiny ID}};    \draw[white, fill=white] (-0.0625,0.095) 
    circle [radius=0.007];    \end{tikzpicture}
    \hspace{-2mm}}
\foreach \x in {A, ..., Z}{%
    \expandafter\xdef\csname orcid\x\endcsname{\noexpand\href{https://orcid.org/\csname orcidauthor\x\endcsname}{\noexpand\orcidicon}}
    }

\usepackage{marvosym}

\begin{document}



\title{Hard-Label Cryptanalytic Extraction of Neural Network Models}

\index{Chen, Yi}
\index{Dong, Xiaoyang}
\index{Guo, Jian}
\index{Shen, Yantian}
\index{Wang, Anyu}
\index{Wang, Xiaoyun}

\titlerunning{Hard-Label Cryptanalytic Extraction of Neural Network Models}
%

\author{Yi Chen\inst{1}\orcidA{} \and
Xiaoyang Dong\inst{2,5}\orcidB{} \and
Jian Guo\inst{3}\orcidC{} \and
Yantian Shen\inst{4}\orcidD{} \and
Anyu~Wang\inst{1,5}\orcidE{} \and
Xiaoyun Wang\inst{1,5,6}\textsuperscript{(\Letter)}\orcidF{}}

%
\authorrunning{Y. Chen et al.}


\institute{
Institute for Advanced Study, 
Tsinghua University, Beijing, China,
\email{chenyi2023@mail.tsinghua.edu.cn},
\email{\{anyuwang, xiaoyunwang\}@tsinghua.edu.cn}
 \\ \and	
Institute for Network Sciences and Cyberspace, BNRist, 
Tsinghua University, Beijing, China,
\email{xiaoyangdong@tsinghua.edu.cn} \\ \and	
School of Physical and Mathematical Sciences,
Nanyang Technological University, Singapore,
\email{guojian@ntu.edu.sg} 	\\ \and	
Department of Computer Science and Technology, 
Tsinghua University, Beijing, China,
\email{shenyt22@mails.tsinghua.edu.cn}  \\ \and	
Zhongguancun Laboratory, Beijing, China  \\ \and 
Shandong Key Laboratory of Artificial Intelligence Security, Shandong, China 
}

\maketitle              
\begin{abstract}

The machine learning problem of 
extracting neural network parameters 
has been proposed for nearly three decades.
Functionally equivalent extraction is a crucial goal
for research on this problem.
When the adversary has access to 
the raw output of neural networks, various attacks, 
including those presented at CRYPTO 2020 and EUROCRYPT 2024, 
have successfully achieved this goal. 
However, this goal is not achieved 
when neural networks operate under a hard-label setting 
where the raw output is inaccessible.

In this paper, 
we propose the first attack that theoretically achieves 
functionally equivalent extraction under the hard-label setting,
which applies to ReLU neural networks.
The effectiveness of our attack is 
validated through practical experiments 
on a wide range of ReLU neural networks,
including neural networks
trained on two real benchmarking datasets
(MNIST, CIFAR10) widely used in computer vision.
For a neural network consisting of $10^5$ parameters,
our attack only requires several hours on a single core.

\end{abstract}
%
%
%


\section{Introduction}
\label{sec:introduction}

Extracting all the parameters (including weights and biases)
of a neural network (called victim model) 
is a long-standing open problem which
is first proposed by cryptographers and mathematicians 
in the early nineties of the last 
century~\cite{DBLP:conf/siemens/BlumR93,
Fefferman1994ReconstructingAN},
and has been widely studied by research groups 
from both industry and 
academia~\cite{DBLP:conf/kdd/LowdM05,
DBLP:journals/csur/OliynykMR23,
DBLP:conf/uss/JagielskiCBKP20,
DBLP:conf/icml/RolnickK20,
DBLP:conf/uss/BatinaBJP19,
DBLP:conf/uss/TramerZJRR16,
DBLP:conf/crypto/CarliniJM20,
DBLP:journals/iacr/CanalesMartinezCHRSS23}.

In previous research~\cite{DBLP:conf/kdd/LowdM05,
DBLP:journals/csur/OliynykMR23,
DBLP:conf/uss/JagielskiCBKP20,
DBLP:conf/icml/RolnickK20,
DBLP:conf/uss/TramerZJRR16,
DBLP:conf/crypto/CarliniJM20,
DBLP:journals/iacr/CanalesMartinezCHRSS23},
one of the most common attack scenarios is as follows.
The victim model (denoted by $f_{\theta}$ 
where $\theta$ denotes the parameters) 
is made available as an Oracle $\mathcal{O}$, 
then the adversary generates inputs $x$
to query $\mathcal{O}$ and collects the feedback $\zeta$
to extract the parameters.
This is similar to a cryptanalysis problem:
$\theta$ is considered the secret key,
and the adversary tries to recover the secret key $\theta$,
given the pairs $\left(x, \zeta \right)$~\cite{DBLP:conf/crypto/CarliniJM20}.
If $f_{\theta}$ contains $m$ parameters 
(64-bit floating-point numbers),
then the secret key $\theta$ contains $64 \times m$ bits,
and the computation complexity of brute force searching
is $2^{64 \times m}$.

Consider that there may be isomorphisms for neural networks,
e.g., permutation and scaling for 
ReLU neural networks~\cite{DBLP:conf/icml/RolnickK20}.
An important concept 
named \emph{functionally equivalent extraction}
is summarized and proposed 
in~\cite{DBLP:conf/uss/JagielskiCBKP20}.

\paragraph{\textup{\textbf{Functionally Equivalent Extraction.  }}}
Denote by $\mathcal{X}$ the input space of 
the victim model $f_{\theta}$.
Functionally equivalent extraction aims at generating
a model $f_{\widehat{\theta}}$ (i.e., the extracted model),
such that $f_{\theta} (x) = f_{\widehat{\theta}} (x)$ 
holds for all $x \in \mathcal{X}$, 
where $f_{\theta} (x)$ and $f_{\widehat{\theta}} (x)$ are,
respectively, the raw output of the victim model and 
the extracted model~\cite{DBLP:conf/uss/JagielskiCBKP20}.
Such extracted model $f_{\widehat{\theta}}$ is called
the functionally equivalent model of $f_{\theta}$
(also say that $f_{\theta}$ and $f_{\widehat{\theta}}$ are
isomorphic~\cite{DBLP:conf/icml/RolnickK20}).
Since $f_{\widehat{\theta}}$ behaves the same as $f_{\theta}$,
the adversary can explore the properties of
$f_{\theta}$ by taking $f_{\widehat{\theta}}$ as a perfect substitute
\footnote[1]{
Due to the isomorphisms
introduced in~\cite{DBLP:conf/icml/RolnickK20}, 
the parameters $\widehat{\theta}$ of the extracted model may be 
different from that $\theta$ of the victim model,
but it does not matter as long as
$f_{\widehat{\theta}}$ is the functionally equivalent model of $f_{\theta}$.
}.

Functionally equivalent extraction is 
hard~\cite{DBLP:conf/uss/JagielskiCBKP20}.
Consider the isomorphisms (permutation and scaling) 
introduced in~\cite{DBLP:conf/icml/RolnickK20}.
Scaling can change parameters and permutation does not. 
For a ReLU neural network $f_{\theta}$ that contains 
$m$ parameters (64-bit floating-point numbers) and $n$ neurons, 
from the perspective of the above cryptanalysis problem, 
even though 
\emph{scaling can be applied to each neuron once},
the adversary still needs to recover $64 \times (m - n)$ secret key bits.
Note that the case of $m \gg n$ is common for neural networks.
For example, for the ReLU neural networks 
extracted by Carlini et al. at 
CRYPTO 2020~\cite{DBLP:conf/crypto/CarliniJM20},
the pairs $(m, n)$ are $(25120, 32)$, $(100480, 128)$,
$(210, 20)$, $(420, 40)$, and $(4020, 60)$.
Besides, even if we only recover a few bits 
(instead of $64$ bits) of each parameter,
the number of secret key bits to be  recovered is still large,
particularly in the case of large $m$.

When a functionally equivalent model 
$f_{\widehat{\theta}}$ is obtained,
the adversary can do more 
damage (e.g., adversarial attack~\cite{DBLP:conf/sp/Carlini017}) 
or steal user privacy (e.g., property inference 
attack~\cite{DBLP:conf/ccs/GanjuWYGB18}).
Thus, even though the cost is expensive,
the adversary has the motivation to
achieve functionally equivalent extraction.


\paragraph{\textup{\textbf{Hard-label Setting.  }}}
According to the taxonomy in~\cite{DBLP:conf/uss/JagielskiCBKP20},
when the Oracle is queried,
there are $5$ types of feedback given by the Oracle:
(1) label (the most likely class label, also called hard-label), 
(2) label and score (the most-likely class label and its probability score), 
(3) top-k scores, (4) all the label scores, 
(5) the raw output (i.e., $f_{\theta} (x)$). 
When the Oracle only returns the hard-label,
we say that the victim model $f_{\theta}$ 
(i.e., neural networks in this paper) 
works under the hard-label setting~\cite{DBLP:conf/ilp/GalstyanC07}.

To the best of our knowledge, 
there are no functionally equivalent extraction attacks
that are based on the first four types of feedback so far.
From the perspective of cryptanalysis, 
the raw output $f_{\theta}(x)$ 
is equivalent to the complete ciphertext 
corresponding to the plaintext (i.e., the query $x$).
The other four types of feedback only reveal
some properties of the ciphertext (raw output $f_{\theta}(x)$).
For example, when $f_{\theta}(x) \in \mathbb{R}$,
the hard-label only tells whether $f_{\theta} (x) > 0$ 
holds or not~\cite{DBLP:conf/ilp/GalstyanC07}.
Among the five types of feedback, 
the raw output leaks the most information, 
while the hard-label (i.e., the first type of feedback) leaks 
the least~\cite{DBLP:conf/uss/JagielskiCBKP20}.

Assuming that the feedback of the Oracle 
is the raw output,
Jagielski et al. propose the first 
functionally equivalent extraction attack against 
ReLU neural networks with 
one hidden layer~\cite{DBLP:conf/uss/JagielskiCBKP20},
which is extended to 
deeper neural networks in~\cite{DBLP:conf/icml/RolnickK20}.
At CRYPTO 2020, Carlini et al. propose 
a differential extraction attack~\cite{DBLP:conf/crypto/CarliniJM20} 
that requires fewer queries than 
the attack in~\cite{DBLP:conf/icml/RolnickK20}.
However, the differential extraction attack 
requires an exponential amount of time,
which is addressed by Canales{-}Mart{\'{\i}}nez et al. 
at EUROCRYPT 
2024~\cite{DBLP:journals/iacr/CanalesMartinezCHRSS23}.
Note that the extraction attacks 
in~\cite{DBLP:conf/icml/RolnickK20,
DBLP:conf/crypto/CarliniJM20,
DBLP:journals/iacr/CanalesMartinezCHRSS23}
are also based on the assumption that 
the feedback is the raw output. 
Due to the dependence on the raw output,
all the authors in~\cite{DBLP:conf/uss/JagielskiCBKP20,
DBLP:conf/icml/RolnickK20,
DBLP:conf/crypto/CarliniJM20,
DBLP:journals/iacr/CanalesMartinezCHRSS23}, state that
the hard-label setting (i.e., the feedback is the first type)
is a defense against functionally equivalent extraction.


The above backgrounds lead to the
question not studied before
\begin{equation}
\begin{array}{c}
\text{\emph{Is it possible to achieve
functionally equivalent extraction against}} \\
\text{\emph{neural network models under the hard-label setting?}} 
\end{array} \nonumber
\end{equation}

\subsection{Our Results and Techniques}

\subsubsection{Results. }
We have addressed this question head-on in this paper.
In total, the answer is yes,
and we propose the first 
functionally equivalent extraction attack against
ReLU neural networks under the hard-label setting.
Here, the definition of functionally equivalent extraction 
proposed in~\cite{DBLP:conf/crypto/CarliniJM20}
is extended reasonably.

\begin{myDef} [\textbf{\textup{Extended Functionally Equivalent Extraction}}]
\label{def:EFEE}
The goal of the \textup{extended functionally equivalent extraction} 
is to generate a model $f_{\widehat{\theta}}$ (i.e., the extracted model),
such that $f_{\widehat{\theta}} (x) = c \times f_{\theta} (x)$ 
holds for all $x \in \mathcal{X}$, 
where $c > 0$ is a fixed constant, 
$f_{\theta} (x)$ and $f_{\widehat{\theta}} (x)$ are,
respectively, the raw output of the victim model and 
the extracted model.
The extracted model $f_{\widehat{\theta}}$ is
the \emph{functionally equivalent model} of
the victim model $f_{\theta}$.
\end{myDef}

Since $f_{\widehat{\theta}} (x) = c \times f_{\theta} (x)$ holds
for all $x \in \mathcal{X}$, i.e., 
$f_{\widehat{\theta}}$ is a simple scalar product of $f_{\theta}$,
the adversary still can explore the properties
of the victim model $f_{\theta}$ by taking 
$f_{\widehat{\theta}}$ as a perfect substitute.
This is why we propose this extended definition.
From the perspective of cryptanalysis,
this extended definition allows the adversary not to guess
the $64$ bits of the constant $c$.
To evaluate the efficacy of our model extraction attacks,
and quantify the degree to 
which a model extraction attack has succeeded in practice,
we generalize the metric named $(\varepsilon, \delta)$-functional equivalence proposed
in~\cite{DBLP:conf/crypto/CarliniJM20}.


\begin{myDef} 
[\textbf{\textup{Extended $(\varepsilon, \delta)$-Functional Equivalence}}]	
\label{def:equivalence-in-practice}
Two models $f_{\widehat{\theta}}$ and $f_{\theta}$
are $(\varepsilon, \delta)$-functional equivalent on $\mathcal{S}$ if
there exists a fixed constant $c > 0$ such that
\begin{equation}
\mathbf{Pr}_{x \in \mathcal{S}} \left[ 
\left| f_{\widehat{\theta}} (x) -  c \times f_{\theta} (x) \right|  
\leqslant \varepsilon
\right]
\geqslant 1 - \delta  \nonumber
\end{equation}	
\end{myDef}

In this paper,
we propose two model extraction attacks,
one of which applies to $0$-deep neural networks,
and the other one applies to $k$-deep neural networks. 
The former attack is the basis of the latter attack.
Our model extraction attacks theoretically 
achieve functionally equivalent extraction
described in Definition~\ref{def:EFEE},
where the constant $c > 0$ is determined 
by the model parameter $\theta$.

We have also performed numerous experiments
on both untrained and trained neural networks,
for verifying the effectiveness of our model extraction attacks in practice.
The untrained neural networks are obtained 
by randomly generating model parameters.
To fully verify our attacks, 
we also adopt two real benchmarking image datasets 
(i.e., MNIST and CIFAR10) widely used in computer vision, 
and train many classifiers (i.e., trained neural networks) 
as the victim model.
Our model extraction attacks 
show good performances in experiments.
The complete experiment results refer to Tables~\ref{tab:attack_on_1_deep_nn}
and~\ref{tab:attack_on_trained_nn} in
Section~\ref{sec:experiments}.
The number of parameters of neural networks
in our experiments is up to $10^5$,
but the runtime of the proposed extraction 
attack on a single core is within several hours.
Our experiment code is uploaded to GitHub
(\url{
https://github.com/AI-Lab-Y/NN_cryptanalytic_extraction
}).

The analysis of the attack complexity is 
presented in
Appendix~\ref{appendix:attack_complexity}.
For the extraction attack on $k$-deep neural networks, 
its query complexity is about 
$\mathcal{O} \left( d_0 \times 2^n 
\times {\rm log}_2^{ \frac{1}{\epsilon} } \right)$,
where $d_0$ and $n$ are, respectively,  
the input dimension (i.e., the size of $x$) 
and the number of neurons,
$\epsilon$ is a precision chosen by the adversary.
The computation complexity is about
$\mathcal{O} \left( n \times 2^{n^2+n+k} \right)$,
where $n$ is the number of neurons and 
$k$ is the number of hidden layers.
The computation complexity of our attack is much lower
than that of brute-force searching.


\paragraph{\textup{\textbf{Techniques. }}}
By introducing two new concepts, 
namely model activation pattern and model signature, 
we obtained some findings as follows.    
A ReLU neural network is composed of 
a certain number of affine transformations 
corresponding to model activation patterns.   
Each affine transformation leaks partial information 
about neural network parameters, 
which is determined by the corresponding model activation pattern.
Most importantly, for a neural network that contains $n$ neurons, 
$n+1$ special model activation patterns 
will leak all the information about the 
neural network parameters.

Inspired by the above findings,
we design a series of methods to 
find decision boundary points, 
recover the corresponding affine transformations, 
and further extract the neural network parameters. 
These methods compose the complete model extraction attacks.

\subsubsection{Organization. }
The basic notations, threat model, attack goal
and assumptions are introduced in Section~\ref{sec:preliminaries}. 
Section~\ref{sec:concepts} introduces some auxiliary concepts.
Then we introduce the overview of our model extraction attacks, 
the idealized model extraction attacks, 
and some refinements in practice in the following three sections respectively. 
Experiments are introduced in Section~\ref{sec:experiments}. 
At last, we present more discussions 
about our work and conclude this paper.


\section{Preliminaries}
\label{sec:preliminaries}

\subsection{Basic Definitions and Notations}

This section presents some necessary definitions and notations.

\begin{myDef} 
[\textbf{\textup{$k$-Deep Neural 
Network~\cite{DBLP:conf/crypto/CarliniJM20}}}]
\label{def:k-deep-nn}
A \textup{$k$-deep neural network} $f_{\theta} (x)$ is a function 
parameterized by $\theta$ that takes inputs from an input space $\mathcal{X}$
and returns values in an output space $\mathcal{Y}$.
The function $f \colon \mathcal{X} \rightarrow \mathcal{Y}$
is composed of alternating linear layers $f_i$ and a non-linear
activation function $\sigma$:
\begin{equation}
f = f_{k+1} \circ \sigma \circ \cdots \circ \sigma \circ f_2 \circ \sigma \circ f_1 .
\end{equation}
\end{myDef}

In this paper, we exclusively study neural networks over
$\mathcal{X} = \mathbb{R}^{d_0}$ and $\mathcal{Y} = \mathbb{R}^{d_{k + 1}}$,
where $d_0$ and $d_{k + 1}$ are positive integers. 
As in~\cite{DBLP:conf/crypto/CarliniJM20,
DBLP:journals/iacr/CanalesMartinezCHRSS23}, 
we only consider neural networks 
using the ReLU~\cite{DBLP:conf/icml/NairH10} activation function,
given by $\sigma \colon x \mapsto \textup{max} (x, 0)$.

\begin{myDef} 
[\textbf{\textup{Fully Connected 
Layer~\cite{DBLP:conf/crypto/CarliniJM20}}}]
The $i$-th \textup{fully connected layer} of a neural network is
a function $f_i \colon \mathbb{R}^{d_{i-1}} \rightarrow \mathbb{R}^{d_i}$
given by the affine transformation
\begin{equation}
f_i (x) = A^{(i)} x + b^{(i)}.
\end{equation}
where $A^{(i)} \in \mathbb{R}^{d_i \times d_{i-1}}$ 
is a $d_i \times d_{i-1}$ \textup{weight} matrix,
$b^{(i)} \in \mathbb{R}^{d_i}$ is a $d_i$-dimensional \textup{bias} vector.
\end{myDef}

\begin{myDef} [\textbf{\textup{Neuron~\cite{DBLP:conf/crypto/CarliniJM20}}}]
A \textup{neuron} is a function determined by the corresponding
weight matrix, bias vector, and activation function.
Formally, the $j$-th neuron of layer $i$ is the function $\eta$ given by
\begin{equation}
\eta (x) = \sigma \left( A_j^{(i)} x + b_j^{(i)} \right),
\end{equation}  
where $A_j^{(i)}$ and $b_j^{(i)}$  denote, respectively, 
the $j$-th row of $A^{(i)}$ and $j$-th coordinate of $b^{(i)}$.
In a $k$-deep neural network, 
there are a total of $\sum_{i=1}^{k} d_i$ neurons.
\end{myDef}

\begin{myDef} [\textbf{\textup{Neuron 
State~\cite{DBLP:journals/iacr/CanalesMartinezCHRSS23}}}]
Let $\mathcal{V}(\eta; x)$ denote the value that neuron $\eta$ takes
with $x \in \mathcal{X}$ before applying $\sigma$. 
If $\mathcal{V}(\eta; x) > 0$, then $\eta$ is active,
i.e., the \textup{neuron state} is \textup{active}.
Otherwise, the neuron state is \textup{inactive}
\footnote[2]{
In~\cite{DBLP:conf/crypto/CarliniJM20,
DBLP:journals/iacr/CanalesMartinezCHRSS23}, 
the authors defined one more neuron state, 
namely critical, i.e., $\mathcal{V}(\eta; x) = 0$,
which is a special inactive state
since the output of neuron $\eta$ is $0$.
}.
The state of the $j$-th neuron in layer $i$ on input $x$ 
is denoted by $\mathcal{P}_j^{(i)} (x) \in \mathbb{F}_2$.
If $\mathcal{P}_j^{(i)} (x) = 1$, the neuron is active.
If $\mathcal{P}_j^{(i)} (x) = 0$, the neuron is inactive.
\end{myDef}

\begin{myDef} 
[\textbf{\textup{Neural Network 
Architecture~\cite{DBLP:conf/crypto/CarliniJM20}}}]
The \textup{architecture} of a fully connected  neural network 
captures the structure of $f_{\theta}$: (a) the number of layers, 
(b) the dimension $d_i$ of each layer $i = 0, \cdots , k+1$.
We say that $d_0$ is the dimension of the input to the neural network,
and $d_{k+1}$ denotes the number of outputs of the neural network. 
\end{myDef}

\begin{myDef} 
[\textbf{\textup{Neural Network 
Parameters~\cite{DBLP:conf/crypto/CarliniJM20}}}]
The \textup{parameters} $\theta$ of a $k$-deep neural network $f_{\theta}$ 
are the concrete assignments
to the weights $A^{(i)}$ and biases $b^{(i)}$ for $i \in \{1, 2, \cdots, k+1\}$.
\end{myDef}


When neural networks work under the hard-label setting,
the raw output $f_{\theta} (x)$ is processed 
before being returned~\cite{DBLP:conf/ilp/GalstyanC07}.
This paper considers the most common processing.
The raw output $f_{\theta} (x) \in \mathbb{R}^{d_{k+1}}$
is first transformed into a category probability vector
$\mathbf{P} \in \mathbb{R}^{d_{k+1}}$ 
by applying the Sigmoid (when $d_{k+1} = 1$) 
or Softmax  (when $d_{k+1} > 1$) function to 
$f_{\theta} (x)$~\cite{DBLP:journals/ijon/DubeySC22}.
Then, the category with the largest probability 
is returned as a hard-label. 
Definition~\ref{def:hard-label} summarizes 
the hard-label and corresponding decision conditions 
on the raw output $f_{\theta}(x)$.

\begin{myDef} [\textbf{\textup{Hard-Label}}]	
\label{def:hard-label}
Consider a $k$-deep neural network 
$f \colon \mathcal{X} \rightarrow \mathcal{Y}$ 
where $\mathcal{Y} \in \mathbb{R}^{d_{k+1}}$.
The \textup{hard-label} (denoted by $z$) is
related to the outputs $f_{\theta} (x)$. 
When $d_{k+1} = 1$, the hard-label $z \left( f_{\theta} (x) \right)$ is computed as
\begin{equation}
z  \left( f_{\theta} (x) \right) = \left\{
\begin{array}{l}
1, \,\, \textup{if} \,\, f_{\theta} (x) > 0 , \\
0, \,\, \textup{if} \,\, f_{\theta} (x) \leqslant 0 .\\
\end{array}
\right.
\end{equation}
When $d_{k+1} > 1$, 
the output $f_{\theta} (x)$ is a $d_{k+1}$-dimensional vector.
The hard-label $z  \left( f_{\theta} (x) \right)$ is 
the coordinate of the maximum of $f_{\theta} (x)$.
\footnote[3]{
If there are ties, i.e., multiple items of $f_{\theta} (x)$ share the same maximum, 
the hard-label is the smallest one of the coordinates of these items.
}
\end{myDef}

\subsection{Adversarial Goals and Assumptions}
\label{subsec:goal_and_assumptions}

There are two parties in a model extraction attack:
the oracle $\mathcal{O}$ who 
possesses the neural network $f_{\theta} (x)$, 
and the adversary who generates queries $x$ to the Oracle. 
Under the hard-label setting, 
the Oracle $\mathcal{O}$ returns the hard-label $z  \left( f_{\theta} (x) \right)$ 
in Definition~\ref{def:hard-label}.

\begin{myDef} [\textbf{\textup{Model Parameter Extraction Attack}}]
A \textup{model parameter extraction attack} 
receives Oracle access to
a parameterized function $f_{\theta}$ 
(i.e., a $k$-deep neural network in our paper)
and the architecture of $f_{\theta}$, 
and returns a set of parameters $\widehat{\theta}$
with the goal that 
$f_{ \widehat{\theta} }(x)$ is as similar as possible to 
$c \times f_{\theta}(x)$ where $c > 0$ is a fixed constant.
\end{myDef}

In this paper,
we use the $\widehat{\_}$ symbol to indicate an extracted parameter.
For example, $\theta$ is the parameters of the victim model $f_{\theta}$,
and $\widehat{\theta}$ stands for the parameters
of the extracted model $f_{\widehat{\theta}}$.

\paragraph{\textup{\textbf{Assumptions. }}}
We make the following assumptions of the Oracle $\mathcal{O}$ 
and the capabilities of the attacker:
\begin{itemize}
\item 
\textbf{Architecture knowledge.  } 
We require knowledge of the neural network architecture.

\item 
\textbf{Full-domain inputs. } 
We can feed arbitrary inputs from $\mathcal{X} = \mathbb{R}^{d_0}$.

\item
\textbf{Precise computations. }
$f_{\theta}$ is specified and evaluated 
using 64-bit floating-point arithmetic.

\item
\textbf{Scalar outputs. }
The output dimensionality is $1$,
i.e., $\mathcal{Y} = \mathbb{R}$.
\footnote[4]{This assumption is fundamental to our work.
Our attack only applies to the case of scalar outputs.}

\item
\textbf{ReLU Activations. }
All activation functions ($\sigma$'s) are the ReLU function.
\end{itemize}

Compared with the work in~\cite{DBLP:conf/crypto/CarliniJM20},
we remove the assumption of requiring 
the raw output $f_{\theta} (x)$ of the neural network.
Now, after querying the Oracle $\mathcal{O}$,
the attacker obtains the hard-label $z\left( f_{\theta} (x) \right)$.
In other words, the attacker only knows whether
$f_{\theta} (x) > 0$ holds or not.


\section{Auxiliary Concepts}
\label{sec:concepts}


To help understand our attacks,
this paper proposes some auxiliary concepts.

\subsection{Model Activation Pattern}



To describe all the neuron states,
we introduce a new concept named
\emph{Model Activation Pattern}.

\begin{myDef} [\textbf{\textup{Model Activation Pattern}}]
\label{def:MAP}
Consider a $k$-deep neural network $f_{\theta}$ 
with $n = \sum_{i=1}^{k}{d_i}$ neurons.
The \textup{model activation pattern} of $f_{\theta}$ 
over an input $x \in \mathcal{X}$ is 
a global description of the $n$ neuron states,
and is denoted by 
$\mathcal{P} (x) = (\mathcal{P}^{(1)} (x), \cdots, \mathcal{P}^{(k)} (x))$
where $\mathcal{P}^{(i)} (x) \in \mathbb{F}_2^{d_i}$
is the concatenation of $d_i$ neuron states 
(i.e., $\mathcal{P}_j^{(i)} (x), i \in \{1, \cdots, d_i\}$) in layer $i$.
\end{myDef}

In the rest of this paper, 
the notations $\mathcal{P}_j^{(i)} (x)$, $\mathcal{P}^{(i)}(x)$,
and $\mathcal{P}(x)$
are simplified as $\mathcal{P}_j^{(i)}$, $\mathcal{P}^{(i)}$,
and $\mathcal{P}$ respectively, 
when the meaning is clear in context.
Besides, $\mathcal{P}^{(i)} \in \mathbb{F}_2^{d_i}$ 
is represented by a $d_i$-bit integer. 
For example, $\mathcal{P}^{(i)} = 2^{j-1}$ means that
only the $j$-th neuron in layer $i$ is active,
and $\mathcal{P}^{(i)} = 2^{d_i} - 1$ means that
all the $d_i$ neurons are active.

When the model activation pattern is known,
one can precisely determine which neural network parameters 
influence the output $f_{\theta} (x)$.
Consider the $j$-th neuron $\eta$ in layer $i$.
Due to the ReLU activation function, 
if the neuron state is \emph{inactive},
neuron $\eta$ does not influence the output $f_{\theta} (x)$. 
As a result, all the weights $A_{?, j}^{(i+1)}$ and $A_{j, ?}^{(i)}$
( i.e., the elements of the $j$-th column of $A^{(i+1)}$,
and the $j$-th row of $A^{(i)}$ respectively)
and the bias $b_j^{(i)}$
do not affect the output $f_{\theta} (x)$.


\paragraph{Special `neuron'. }
For the convenience of introducing model extraction attacks later,
we regard the input $x \in \mathbb{R}^{d_0}$
and the output $f_{\theta} (x) \in \mathbb{R}$ as, 
respectively, $d_0$ and $1$ special `neurons' that are always active.
So we adopt two extra notations $\mathcal{P}^{(0)} = 2^{d_0} - 1$
and $\mathcal{P}^{(k+1)} = 2^1 - 1 = 1$, for describing the states
of the special $d_0 + 1$ `neurons'.
But if not necessary, we will omit the two notations.

\subsection{Model Signature}	\label{subsec:model_signature}
Consider a $k$-deep neural network $f_{\theta}$.
For an input $x \in \mathcal{X}$, 
$f_{\theta}$ can be described as an affine transformation
\begin{equation}
\begin{aligned}
f_{\theta} (x) &=  A^{(k+1)} \cdots 
          \left(  I_{\mathcal{P}}^{(2)} 
          \left( A^{(2)} 
          \left( I_{\mathcal{P}}^{(1)} 
          \left( A^{(1)} x + b^{(1)}
          \right) 
          \right) + b^{(2)} 
          \right)  
          \right) 
          \cdots  + b^{(k+1)}	\\
     &= A^{(k+1)}  I_{\mathcal{P}}^{(k)} A^{(k)} \cdots I_{\mathcal{P}}^{(2)}
           A^{(2)} I_{\mathcal{P}}^{(1)} A^{(1)} x + B_{\mathcal{P}}   
        = \varGamma _{\mathcal{P}} x + B_{\mathcal{P}}  ,
\end{aligned}
\end{equation}
where $\mathcal{P}$ is the model activation pattern over $x$,
$\varGamma _{\mathcal{P}} \in \mathbb{R}^{d_0}$, 
and $B_{\mathcal{P}} \in \mathbb{R}$.
Here, $I_{\mathcal{P}}^{(i)} \in \mathbb{R}^{d_i \times d_i}$ 
are $0$-$1$ diagonal matrices with
a $0$ on the diagonal's $j$-th entry 
when the neuron state $\mathcal{P}_{j}^{(i)}$ is $0$,
and $1$ when $\mathcal{P}_{j}^{(i)} = 1$.

The affine transformation is denoted by a tuple 
$(\varGamma _{\mathcal{P}}, B_{\mathcal{P}})$.
Except for $\mathcal{P}$,
the value of the tuple 
$(\varGamma _{\mathcal{P}}, B_{\mathcal{P}})$
is only determined by the neural network parameters, 
i.e., $A^{(i)}$ and $b^{(i)}, i \in \{1, \cdots, k+1\}$.
Once the value of any neural network parameters is changed,
the value of the tuple $(\varGamma _{\mathcal{P}}, B_{\mathcal{P}})$
corresponding to some $\mathcal{P}$'s will change too
\footnote[5]{
We do not consider the case of some neurons being always inactive, 
since such neurons are redundant and 
usually deleted by various network pruning methods 
(e.g.,~\cite{DBLP:conf/nips/HanPTD15})
before the neural network is deployed as a prediction service.
}.
Therefore, we regard the set of all the possible
tuples $(\varGamma _{\mathcal{P}}, B_{\mathcal{P}})$
as a unique \emph{model signature} of the neural network.

\begin{myDef} [\textbf{\textup{Model Signature}}]
For a $k$-deep neural network $f_{\theta} (x)$,
the \textup{model signature} denoted by $\mathcal{S}_{\theta}$
is the set of affine transformations
\begin{equation}
\mathcal{S}_{\theta} = 
\{(\varGamma _\mathcal{P},  B_\mathcal{P}) 
\text{ for all the } \mathcal{P}\text{'s} \}  .
\nonumber
\end{equation}
\end{myDef}

In~\cite{DBLP:journals/iacr/CanalesMartinezCHRSS23},
Canales{-}Mart{\'{\i}}nez et al. use the term `signature'
to describe the weights related to a neuron,
which is different from the model signature.
Except for the model signature,
we propose another important concept, 
namely \emph{normalized model signature}.

\begin{myDef} [\textbf{\textup{Normalized Model Signature}}]
\label{def:nms}
Consider a victim model $f_{\theta}$ and its model signature
$\mathcal{S}_{\theta} = 
\{(\varGamma _\mathcal{P},  B_\mathcal{P}) 
\text{ for all the } \mathcal{P}\text{'s} \}$ .  
Denote by $\varGamma_{\mathcal{P}, j}$ the $j$-th element of $\varGamma_{\mathcal{P}}$ for $j \in \{1, \cdots, d_0\}$.
Divide the set of $\mathcal{P}$'s into two subsets 
$\mathcal{Q}_1$ and $\mathcal{Q}_2$.
For each $\mathcal{P} \in \mathcal{Q}_1$, 
$\varGamma_{\mathcal{P}, j} = 0$ for $j \in \{1, \cdots, d_0\}$.
For each $\mathcal{P} \in \mathcal{Q}_2$, 
there is at least one non-zero element in $\varGamma_{\mathcal{P}}$,
without loss of generality,  
assume that $\varGamma_{\mathcal{P}, 1} \ne 0$.
Let $\mathcal{S}_{\theta}^{\mathcal{N}}$ be the following set
\begin{equation}
\mathcal{S}_{\theta}^{\mathcal{N}} = 
\left\{ 
\left(
\varGamma _\mathcal{P},  
B_\mathcal{P}
\right) 
\text{ for } \mathcal{P} \in \mathcal{Q}_1,
\left(
\frac{
\varGamma _\mathcal{P}
}{
\left| \varGamma_{\mathcal{P}, 1} \right|
},  
\frac{
B_\mathcal{P}
}{
\left| \varGamma_{\mathcal{P}, 1} \right|
} \right) 
\text{ for } \mathcal{P} \in \mathcal{Q}_2
\right\} .  	\nonumber
\end{equation}
The set $\mathcal{S}_{\theta}^{\mathcal{N}}$ 
is the \textup{normalized model signature} of $f_{\theta}$.
\end{myDef}

Shortly, the difference between the normalized model signature
$\mathcal{S}_{\theta}^{\mathcal{N}}$
and the initial model signature $\mathcal{S}_{\theta}$ is as follows.
For each $\mathcal{P} \in \mathcal{Q}_2$, 
i.e., there is at least one non-zero element
in $\varGamma_{\mathcal{P}}$ (without loss of generality, 
assume that the first element is non-zero, 
i.e.,$\varGamma_{\mathcal{P}, 1} \ne 0$),
we transform the parameter tuple into 
$\left(
\frac{
\varGamma _\mathcal{P}
}{
\left| \varGamma_{\mathcal{P}, 1} \right|
},  
\frac{
B_\mathcal{P}
}{
\left| \varGamma_{\mathcal{P}, 1} \right|
} \right)$.

In our attacks, 
the normalized model signature plays two important roles.
First, the recovery of all the weights $A^{(i)}$
relies on the subset $\mathcal{Q}_2$.
Second,  our attacks will produce many extracted models 
during the attack process while at most only one 
is the functionally equivalent model of $f_{\theta}$,
and the normalized model signature is used to filter
functionally inequivalent models.

\subsection{Decision Boundary Point}
Our attacks exploit a special class of inputs 
named \emph{Decision Boundary Points}.

\begin{myDef} [\textbf{\textup{Decision Boundary Point}}]
\label{def:DBP}
Consider a neural network $f_{\theta}$.
If an input $x$ makes $f_{\theta} (x) = 0$ hold,
$x$ is a \textup{decision boundary point}.
\end{myDef}

The extraction attacks presented
at CRYPTO 2020~\cite{DBLP:conf/crypto/CarliniJM20} and EUROCRYPT 2024~\cite{DBLP:journals/iacr/CanalesMartinezCHRSS23}
exploit a class of inputs, namely critical points.
Fig.~\ref{fig:DBP_critical_point} shows the difference between
critical points and decision boundary points.

\begin{figure}[htb]
\centering
\includegraphics[width=0.6\textwidth]{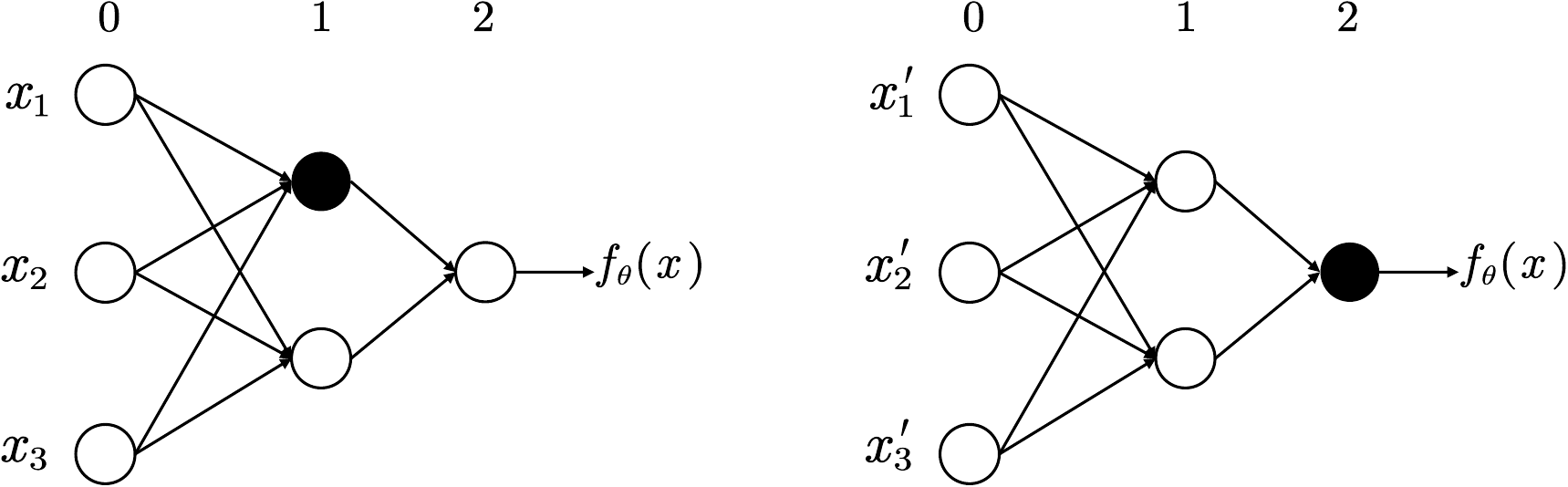}
\caption{
Left: the critical point $x = [x_1, x_2, x_3]^{\top}$ 
makes the output of one neuron (e.g., the solid black circle) $0$.
Right: the decision boundary point 
$x^{\prime} = [x_1^{\prime}, x_2^{\prime}, x_3^{\prime}]^{\top}$ 
makes the output of the neural network $0$.
}
\label{fig:DBP_critical_point}
\end{figure}

Critical points leak information on the neuron states,
i.e., whether the output of a neuron is $0$,
which is the core reason why the differential extraction attack
can efficiently extract model 
parameters~\cite{DBLP:conf/crypto/CarliniJM20}.
As a comparison, decision boundary points 
do not leak information on the neuron states.

Finding critical points relies on computing
partial derivatives based on the raw output $f_{\theta} (x)$,
refer to the work in~\cite{DBLP:conf/crypto/CarliniJM20,
DBLP:journals/iacr/CanalesMartinezCHRSS23}.
Thus, under the hard-label setting, 
we can not exploit critical points.

\section{Overview of Our Cryptanalytic Extraction Attacks}
\label{sec:overview}

Under the hard-label setting, 
i.e., the Oracle returns the most likely class 
$z \left( f_{\theta} (x) \right)$ instead of the raw output $f_{\theta}(x)$, 
only decision boundary points $x$ 
will leak the value of $f_{\theta}(x)$,
since $f_{\theta} (x) = 0$.
Motivated by this truth,
our cryptanalytic extraction attacks 
focus on decision boundary points.

\paragraph{\textup{\textbf{Attack Process.  }}}

At a high level, 
the complete attack contains five steps.
\begin{enumerate}
\item [$\cdot$]
\textbf{Step 1: collect decision boundary points. } 
The algorithm for finding decision boundary points
will be introduced in Section~\ref{subsec:find_DBP}.
Suppose that $M$ decision boundary points are collected.

\item [$\cdot$]
\textbf{Step 2: recover the normalized model signature. }
Recover the tuples 
$\left( \varGamma_{\mathcal{P}}, B_{\mathcal{P}} \right)$
corresponding to the $M$ decision boundary points.
After filtering duplicate tuples,
regard the set of the remaining tuples as
the (partial)
normalized model signature
$\mathcal{S}_{\theta}^{\mathcal{N}}$.
Suppose that the size of $\mathcal{Q}_2$ is $N$,
refer to Definition~\ref{def:nms}.
It means that  there are $N$ decision boundary points 
that can be used to recover weights $A^{(i)}$.


\item [$\cdot$]
\textbf{Step 3: recover weights layer by layer. }
Suppose that there are $n = \sum_{i=1}^{k} d_i$ 
neurons in the neural network.
Randomly choose $n + 1$ out of $N$ decision boundary points each time,
assign a specific model activation pattern $\mathcal{P}$ 
to each selected decision boundary point,
and recover the weights 
$A^{(1)}, \cdots, A^{(k+1)}$. 

\item [$\cdot$]
\textbf{Step 4: recover all the biases. } 
Based on recovered weights, 
recover all the biases $b^{(i)}, i \in \{1, \cdots, k+1\}$
simultaneously. 

\item [$\cdot$]
\textbf{Step 5: filter functionally inequivalent models. }
As long as $N \geqslant n + 1$ holds, 
we will obtain many extracted models,
but it is expected that at most only one is the functionally equivalent model.
Thus, we filter functionally inequivalent models in this step.
\end{enumerate}
Some functionally inequivalent models may not be filtered.
For each surviving extracted model,
we test the \emph{Prediction Matching Ratio} 
(PMR, introduced in Section~\ref{subsec:filter_wrong_models}) 
over randomly generated inputs, 
and take the one with the highest PMR as the final candidate.

In Step 2, 
we recover the tuple 
$\left( \varGamma_{\mathcal{P}}, B_{\mathcal{P}} \right)$
by the extraction attack on $0$-deep neural networks.
In Step 3, for layer $i > 1$, 
the weight vector $A_j^{(i)}$
of the $j$-th neuron ($j \in \{1, \cdots, d_i\}$) 
is recovered by solving a system of linear equations.
For layer $1$, except for selecting $d_1$ decision boundary points,
the recovery of the weights $A^{(1)}$
does not use any extra techniques.
In Step 4, all the biases are recovered 
by solving a system of linear equations.

\section{Idealized Hard-Label Model Extraction Attack}
\label{sec:DL_attack}

This section introduces $(0, 0)$-functionally equivalent
model extraction attacks under the hard-label setting, 
which assumes infinite precision arithmetic and 
recovers the functionally equivalent model.
We first introduce the 
$0$-deep neural network extraction attack,
which is used in the $k$-deep neural network extraction attack
to recover the normalized model signature.

Note that this section only (partially) involves Steps 2, 3, and 4 
introduced in Section~\ref{sec:overview}.
In the next section, 
we introduce the remaining steps
and refine the idealized attacks to work with finite precision.

%

\subsection{Zero-Deep Neural Network Extraction}
\label{subsec:0_nn_attack}

According to Definition~\ref{def:k-deep-nn},
zero-deep neural networks are affine functions
$f_{\theta}(x) \equiv A^{(1)} \cdot x + b^{(1)}$
where $A^{(1)} \in \mathbb{R}^{d_0}$,
and $b^{(1)} \in \mathbb{R}$. 
Let $A^{(1)} = [w_1^{(1)}, \cdots, w_{d_0}^{(1)}]$,
and $x = [x_1, x_2, \cdots, x_{d_0}]^\top$.
The model signature is 
$\mathcal{S}_{\theta} = \left( A^{(1)},  b^{(1)} \right)$.

Our extraction attack
is based on a decision boundary point 
$x$ (i.e., $f_{\theta} (x) = 0$),
and composed of $3$ steps:
(1) recover weight signs, i.e., the sign of $w_i^{(1)}$;
(2) recover weights $w_{i}^{(1)}$;
(3) recover bias $b^{(1)}$.

\paragraph{\textup{\textbf{Recover Weight Signs. }}}
Denote by $e_i \in \mathbb{R}^{d_0}$ 
the basis vector where only the $i$-th element is $1$ 
and other elements are $0$.

Let the decision boundary point $x$ 
move along the direction $e_i, i \in \{1, \cdots, d_0\}$,
and the moving stride is $s \in \mathbb{R}$
where $|s| > 0$.
Query the Oracle and 
obtain the hard-label $z\left( f (x + s e_i) \right)$,
then the sign of $w_i^{(1)}$ is
\begin{equation}
{\rm sign} (w_i^{(1)}) = \left\{ 
\begin{array}{c}
\,\,\,\,1, \,\, {\rm if} \,\, s>0 \,\, {\rm and} \,\,
                           z \left(f_{\theta} (x + s e_i) \right) = 1 , \\
-1, \,\, {\rm if} \,\, s<0 \,\, {\rm and} \,\,
                           z \left(f_{\theta} (x + s e_i) \right) = 1 . \\
\end{array}
\right.
\end{equation}
When $z\left( f(x + s e_i) \right) = 1$,
we have $f (x + s e_i) > 0$,
i.e., $w_i^{(1)} \times s > 0$.
Thus, the sign of $w_i^1$ is the same as that of $s$.
If $z\left( f(x + s e_i) \right) = 0$ always holds,
no matter if $s$ is positive or negative,
then we have $w_i^{(1)} = 0$.

\paragraph{\textup{\textbf{Recover Weights. }}}
Without loss of generality, assume that $w_1^{(1)} \ne 0$.

At first, let the decision boundary point $x$ move along 
$e_1$ with a moving stride $s_1$, such that
the hard-label of the new point $x + s_1 e_1$ is $1$,
i.e., $z \left( f(x + s_1 e_1) \right) = 1$.
Then, let the new point $x + s_1 e_1$ 
move along $e_i$ with a moving stride $s_i$
where $i \ne 1$ and $w_i^{(1)} \ne 0$ , 
such that  $x + s_1 e_1 + s_i e_i$ 
is a decision boundary point too.
As a result, we have 
\begin{equation}
s_1 w_1^{(1)} + s_i w_i^{(1)} = 0,
\end{equation}
and obtain the weight ratio $\frac{w_i^{(1)}}{w_1^{(1)}}$.
Since the signs of $w_i^{(1)}$ are known,
the final extracted weights are
\begin{equation}
\widehat{A}^{(1)} = \left[ 
\frac{w_1^{(1)}}{\left| w_1^{(1)} \right|}, 
\frac{w_2^{(1)}}{\left| w_1^{(1)} \right|}, \cdots, 
\frac{w_{d_0}^{(1)}}{\left| w_1^{(1)} \right|} \right]  .
\end{equation}


\paragraph{\textup{\textbf{Recover Bias. }}}
The extracted bias is 
$\widehat{b}^{(1)} = - \widehat{A}^{(1)} \cdot x 
                               = \frac{b^{(1)}}{\left| w_1^{(1)} \right|}$ .

Thus, the model signature of $f_{\widehat{\theta}}$ is
$\mathcal{S}_{\widehat{\theta}} = 
\left(\frac{ A^{(1)}}{\left| w_1^{(1)} \right|},  
\frac{b^{(1)}}{\left| w_1^{(1)} \right|} \right)$,
and $f_{\widehat{\theta}} (x) = \frac{f(x)}{\left| w_1^{(1)} \right|}$.

\begin{remark}
In~\cite{DBLP:conf/kdd/LowdM05},
the authors propose different methods to extract
the parameters of linear functions 
$f_{\theta} (x) = A^{(1)} \cdot x$.
Since this paper mainly focuses on 
the extraction attack on $k$-deep neural networks,
we do not deeply compare our attack with the methods 
in~\cite{DBLP:conf/kdd/LowdM05}.
\end{remark}

\subsection{$k$-Deep Neural Network Extraction}
\label{subsec:k_deep_nn_attack}

%
%
%

Basing the $0$-deep neural network extraction attack, 
we develop an extraction attack on $k$-deep neural networks.
Recall that, the expression of $k$-deep neural networks is
\begin{equation}	\label{eq:k-nn-form}
\begin{aligned}
f_{\theta}(x) &= A^{(k+1)} \cdots 
          \left(  I_{\mathcal{P}}^{(2)} 
          \left( A^{(2)} 
          \left( I_{\mathcal{P}}^{(1)} 
          \left( A^{(1)} x + b^{(1)}
          \right) 
          \right) + b^{(2)} 
          \right)  
          \right) 
          \cdots + b^{(k+1)}	\\
     &= \varGamma _{\mathcal{P}} x + B_{\mathcal{P}}  
\end{aligned}	
\end{equation}
where the model activation pattern is
$\mathcal{P} = 
\left(  
\mathcal{P}^{(0)},
\mathcal{P}^{(1)}, \cdots, \mathcal{P}^{(k)}, \mathcal{P}^{(k+1)}
\right)$ and
\begin{equation}
\varGamma_{\mathcal{P}} = 
     A^{(k+1)}  I_{\mathcal{P}}^{(k)} A^{(k)} \cdots I_{\mathcal{P}}^{(2)}
     A^{(2)} I_{\mathcal{P}}^{(1)} A^{(1)}  .
\end{equation}

\paragraph{Notations. }
Our attack recovers weights layer by layer.
Assuming that the weights of the first $i-1$ layers
have been recovered and we are trying to recover $A^{(i)}$
where $i \in \{1,  \cdots, k+1\}$, 
we describe $k$-deep
neural networks as:
\begin{equation}	\label{eq:f_recovered_and_unrecovered}
f_{\theta} (x) = \varGamma _{\mathcal{P}} x + B_{\mathcal{P}} = 
		      \mathcal{G}^{(i)} A^{(i)} 
                        C^{(i-1)} x + B_{\mathcal{P}},
\end{equation}
where $\mathcal{G}^{(i)} \in \mathbb{R}^{d_i}$
and $C^{(i-1)} \in \mathbb{R}^{d_{i-1} \times d_0}$ are,
respectively, 
related to the unrecovered part (excluding $A^{(i)}$)
and recovered part of the neural network $f_{\theta}$.

The values of $\mathcal{G}^{(i)}$ and $C^{(i-1)}$ are
\begin{equation}	\label{eq:G_C}
\begin{aligned}
\mathcal{G}^{(i)} &= 
\left\{
\begin{array}{l} 
A^{(k+1)} I_{\mathcal{P}}^{(k)} A^{(k)} \cdots 
    I_{\mathcal{P}}^{(i+1)} A^{(i+1)} I_{\mathcal{P}}^{(i)}  , 
	\,  \text{if} \,\, i \in \{1, \cdots, k\}	\\
1,  \,  \text{if} \,\,  i = k + 1	\\ 
\end{array}
\right.	\\
C^{(i-1)} &= 
\left\{
\begin{array}{l}
I, \,   \text{if} \,\, i = 1  \\
I_{\mathcal{P}}^{(i-1)} A^{(i-1)} \cdots I_{\mathcal{P}}^{(1)} A^{(1)} ,
	\,    \text{if} \,\,  i \in \{2,  \cdots, k+1\}		\\
\end{array}
\right.
\end{aligned}
\end{equation}
where $C^{(0)} = I \in \mathbb{R}^{d_0 \times d_0}$ 
is a diagonal matrix with a $1$ on each diagonal entry.

\paragraph{\textup{\textbf{Core Idea of Recovering Weights Layer by Layer.   }}}
To better grasp the attack details 
presented later,
we first introduce the core idea of 
recovering weights layer by layer.
Assuming that the extracted weights of 
the first $i-1$ layers are known, i.e., 
$\widehat{A}^{(1)}, \cdots, \widehat{A}^{(i-1)}$ are known,
we try to recover the weights in layer $i$.

To obtain the weight vector
$\widehat{A}_j^{(i)}$ of the $j$-th neuron 
(denoted by $\eta_j^{(i)}$) 
in layer $i \in \{1, \cdots, k+1\}$
\footnote[6]{
when $i = k+1$, it means that we are trying to recover
the weights $A^{(k+1)}$.
},
we exploit a decision boundary point 
with the model activation pattern
$\mathcal{P} = \left( \mathcal{P}^{(0)}, \mathcal{P}^{(1)}, 
\cdots, \mathcal{P}^{(k)}, \mathcal{P}^{(k+1)} \right)$ where
\begin{equation}
\mathcal{P}^{(i-1)} = 2^{d_{i-1}} - 1,  \,\,
\mathcal{P}^{(i)} = 2^{j-1} .
\end{equation}
It means that,
in layer $i$, only the $j$-th neuron is active,
and all the $d_{i-1}$ neurons in layer $i-1$ are active.
Fig~\ref{fig:overview_of_attack} shows a schematic diagram
under this scenario.

\begin{figure}[htb]
\centering
\includegraphics[width=0.9\textwidth]{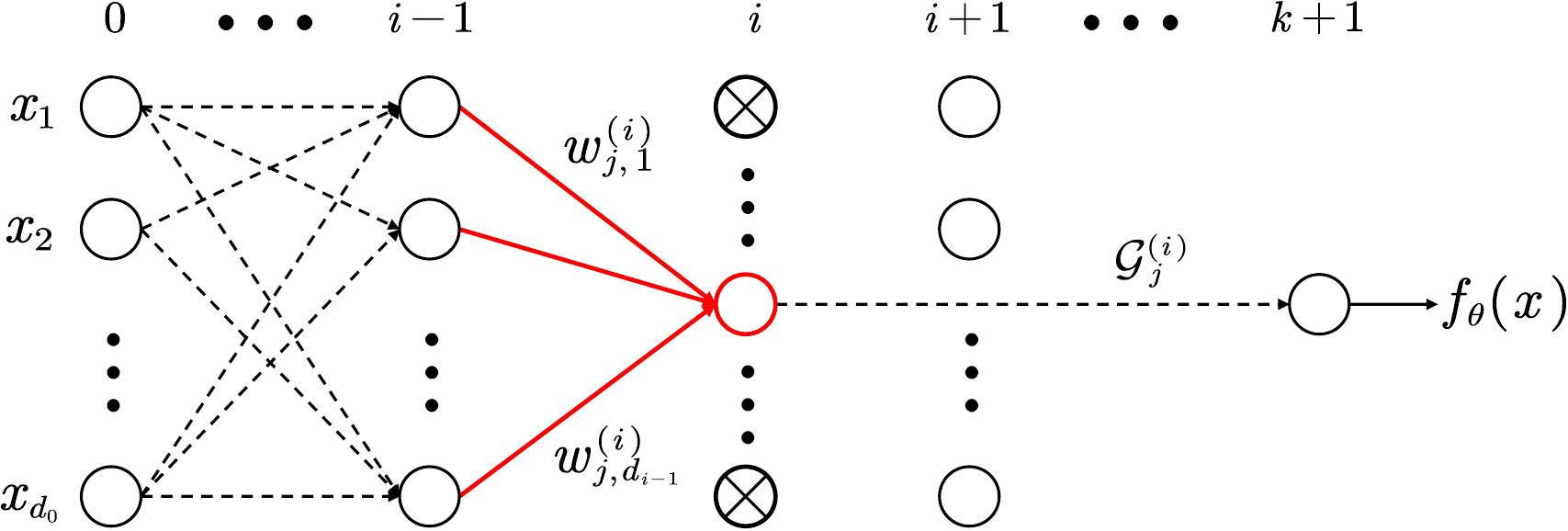}
\caption{
The core idea of recovering the weight vector 
of the $j$-th neuron in layer $i$.
Let $x = [x_1, \cdots, x_{d_0}]^\top$ be a decision boundary point
with $\mathcal{P}^{(i-1)} = 2^{d_{i-1}} - 1$,
$\mathcal{P}^{(i)} = 2^{j-1}$, i.e.,
in layer $i$, only the $j$-th neuron (the red hollow circle) is active,
and in layer $i-1$, all the neurons are active.
The first $i-1$ layers have been extracted,
and collapse into one layer.
All the layers starting from layer $i+1$ collapse into
a direct connection between 
the $j$-th neuron in layer $i$ and the final output.
}
\label{fig:overview_of_attack}
\end{figure}

Since $\mathcal{P}^{(i)} = 2^{j-1}$,
all the $k-i$ layers starting from layer $i+1$
collapse into a direct connection 
from $\eta_j^{(i)}$ to the output $f_{\theta} (x)$.
The weight of this connection is
$\mathcal{G}_j^{(i)}$, i.e., 
the $j$-th element of 
$\mathcal{G}^{(i)}$ (see Eq.~\eqref{eq:G_C}).
The expression (see Eq.~\eqref{eq:f_recovered_and_unrecovered}) 
of the $k$-deep neural network further becomes
\begin{equation}
f_{\theta} (x) = \varGamma _{\mathcal{P}} \cdot x + B_{\mathcal{P}}
= \mathcal{G}_j^{(i)} A_j^{(i)} \cdot C^{(i-1)} \cdot x + B_{\mathcal{P}},
\nonumber
\end{equation}
where $\mathcal{G}_j^{(i)} \in \mathbb{R}$
and $A_j^{(i)} \cdot C^{(i-1)} \in \mathbb{R}^{d_0}$.

In Step 2 (see Section~\ref{sec:overview}), 
applying the extraction attack on zero-deep neural networks,
we can obtain the tuple
$\left(
\frac{
\varGamma _\mathcal{P}
}{
\left| \varGamma_{\mathcal{P}, 1} \right|
},  
\frac{
B_\mathcal{P}
}{
\left| \varGamma_{\mathcal{P}, 1} \right|
} \right)$ where
\begin{equation}	\label{eq:gamma_p_v}
\varGamma _\mathcal{P} = 
\mathcal{G}_j^{(i)} A_j^{(i)} \cdot C^{(i-1)}, \,\,
\varGamma_{\mathcal{P}, v} =
\mathcal{G}_j^{(i)} A_j^{(i)} \cdot C_{?, v}^{(i-1)}.
\end{equation}
Here the symbol $C_{?, v}^{(i-1)}$ stands for
the $v$-th column vector of $C^{(i-1)}$. 

According to Eq.~\eqref{eq:gamma_p_v},
the value of each element of 
$\frac{
\varGamma _\mathcal{P}
}{
\left| \varGamma_{\mathcal{P}, 1} \right|
}$ is not related to the absolute value of $\mathcal{G}_j^{(i)}$,
i.e., the unrecovered part does not affect the affine transformation.
Then, basing the vector $\frac{
\varGamma _\mathcal{P}
}{
\left| \varGamma_{\mathcal{P}, 1} \right|
}$ and the extracted weights
$\widehat{A}^{(1)}, \cdots, \widehat{A}^{(i-1)}$,
we build a system of linear equations 
and solve it to obtain $\widehat{A}_j^{(i)}$.
Next, we introduce more attack details.





\paragraph{\textup{\textbf{Recover Weights in Layer $1$.  }}}
To recover the weight vector of 
the $j$-th neuron  in layer $1$,
we exploit the model activation pattern $\mathcal{P}$ where
\begin{equation}	\label{eq:MAPs_for_layer_1}
\mathcal{P}^{(1)} = 2^{j -1}; \,\,
\mathcal{P}^{(i)} = 2^{d_i} - 1, \,\, 
{\rm for} \,\, i \in \{0, 2, 3, \cdots, k+1\}.
\end{equation}
It means that,
in layer $1$, only the $j$-th neuron is active,
and all the neurons in other layers are active.


Under this model activation pattern,
according to Eq.~\eqref{eq:G_C}, we have
\begin{equation}
\mathcal{G}^{(1)} = A^{(k+1)} A^{(k)} \cdots 
                                 A^{(2)} I_{\mathcal{P}}^{(1)}  .
\end{equation}
Now, the expression of 
the $k$-deep neural network is
\begin{equation}
f_{\theta} (x) = 
      \mathcal{G}_j^{(1)}
     \left( A_j^{(1)} x +  b_j^{(1)}\right) + 
     B_{\left( \mathcal{P}^{(2)}, \cdots, \mathcal{P}^{(k)} \right)} 
= \mathcal{G}_j^{(1)} A_j^{(1)} x + B_{\mathcal{P}}  
\end{equation}
where 
$A_j^{(1)} = \left[ w_{j,1}^{(1)}, \cdots, w_{j, d_0}^{(1)} \right]$, 
$\mathcal{G}_j^{(1)} \in \mathbb{R}$ is the $j$-th element 
of $\mathcal{G}^{(1)}$.
As for $B_{\left( \mathcal{P}^{(2)}, \cdots, 
\mathcal{P}^{(k)} \right)} \in \mathbb{R}$,
it is a constant determined by 
$\left( \mathcal{P}^{(2)}, \cdots, \mathcal{P}^{(k)} \right)$.
In other words,
when $\mathcal{P}^{(1)}$ changes,
the value of 
$B_{\left( \mathcal{P}^{(2)}, \cdots, \mathcal{P}^{(k)} \right)}$
does not change.

Recall that, in Step 2, we have recovered the following weight vector
\begin{equation}
\frac{
\varGamma_{\mathcal{P}}
}{
\left| \varGamma_{\mathcal{P}, 1} \right|
} = 
\left[
\frac{\mathcal{G}_j^{(1)} w_{j,1}^{(1)}}
{\left|  \mathcal{G}_j^{(1)} w_{j,1}^{(1)} \right|}, 
\cdots,
\frac{\mathcal{G}_j^{(1)} w_{j,d_0}^{(1)}}
{\left|  \mathcal{G}_j^{(1)} w_{j,1}^{(1)} \right|}
\right],   j \in \{1, \cdots, d_1\} .
\end{equation}
In this step, our target is to obtain
$\widehat{A}_{j}^{(1)}$ where
\begin{equation}	\label{eq:k-attack-step-2-target}
\widehat{A}_j^{(1)} = 
\left[ \widehat{w}_{j,1}^{1}, \cdots, \widehat{w}_{j,d_0}^{1} \right] =
\left[ \frac{w_{j,1}^{1}}{ \left|w_{j,1}^{1} \right|},
\cdots,
         \frac{w_{j,d_0}^{1}}{ \left|w_{j,1}^{1} \right|} \right] ,
 j \in \{1, \cdots, d_1\}.
\end{equation}
Therefore, we need to determine $d_1$ signs,
i.e., the signs of 
$\mathcal{G}_j^{(1)}$
for $\mathcal{P}^{(1)} = 2^{j-1}$ 
where $j \in \{1, \cdots, d_1\}$.

Since $A_j^{(1)} x + b_j^{(1)} > 0$,
we know that
$\mathcal{G}_j^{(1)}
\times B_{\left( \mathcal{P}^{(2)}, \cdots, \mathcal{P}^{(k)} \right)} < 0$ 
holds for $j \in \{1, \cdots, d_1\}$,
which tells us that 
the above $d_1$ signs are the \emph{same}.
Thus, by guessing $1$ sign, 
i.e., the sign of $\mathcal{G}_j^{(1)}$ for
$\mathcal{P}^{(1)} \in \{2^{1-1}, \cdots, 2^{d_1-1}\}$,
we obtain $d_1$ weight vectors 
presented in Eq.~\eqref{eq:k-attack-step-2-target}.

\paragraph{\textup{\textbf{Recover Weights in Layer $i$ ($i > 1$).  }}}
To recover the weight vector of the $j$-th neuron in layer $i$,
we exploit the model activation pattern $\mathcal{P}$ where
\begin{equation}	\label{eq:MAPs_for_layer_i}
\mathcal{P}^{(i)} = 2^{j - 1}; \,\,
\mathcal{P}^{(q)} = 2^{d_q} - 1, \,\, 
{\rm for} \,\, q \in \{0, \cdots, i-1, i+1, \cdots, k+1\}.
\end{equation}
It means that, 
in layer $i$, only the $j$-th neuron is active,
and all the neurons in other layers are active.


Under this model activation pattern,
according to Eq.~\eqref{eq:G_C}, we have
\begin{equation}	\label{eq:current_G_C}
\begin{aligned}
\mathcal{G}^{(i)} &= 
\left\{
\begin{array}{l} 
A^{(k+1)} A^{(k)} \cdots A^{(i+2)} A^{(i+1)} I_{\mathcal{P}}^{(i)},
	\,  \text{if} \,\, i \in \{2, \cdots, k\}	,\\
1,  \,  \text{if} \,\,  i = k + 1	,\\ 
\end{array}
\right.	\\
C^{(i-1)} &= 
	A^{(i-1)} A^{(i-2)}  \cdots A^{(1)} ,
	\,    \text{if} \,\,  i \in \{2,  \cdots, k+1\}	.	\\
\end{aligned}
\end{equation}
Now,
the expression of $k$-deep neural networks becomes
\begin{equation}	\label{eq:k-deep-nn-form}
\begin{aligned}
f_{\theta} (x) &= 
        \mathcal{G}_j^{(i)}
	\left(
        A_j^{(i)}  C^{(i-1)}  x +
	B_{\left( \mathcal{P}^{(1)}, \cdots, \mathcal{P}^{(i)} \right)}
	\right)  +
        B_{\left( \mathcal{P}^{(i+1)}, \cdots, \mathcal{P}^{(k)} \right)}  \\
&=  \mathcal{G}_j^{(i)}  A_j^{(i)} C^{(i-1)}  x +
       B_{\mathcal{P}} ,
\end{aligned}
\end{equation}
where $A_{j}^{(i)} = \left[ w_{j,1}^{(i)}, 
            \cdots, w_{j, d_{i-1}}^{(i)}  \right]$, 
$\mathcal{G}_j^{(i)} \in \mathbb{R}$ and
$C^{(i-1)} \in \mathbb{R}^{d_{i-1} \times d_0}$.
Besides,  $B_{\left( \mathcal{P}^{(i+1)}, \cdots, 
\mathcal{P}^{(k)} \right)} \in \mathbb{R}$ is not related to 
$\left( \mathcal{P}^{(1)}, \cdots, \mathcal{P}^{(i)} \right)$,
and only determined by 
$\left( \mathcal{P}^{(i+1)}, \cdots, \mathcal{P}^{(k)} \right)$,
i.e.,  
$B_{\left( \mathcal{P}^{(i+1)}, \cdots, 
\mathcal{P}^{(k)} \right)} $ is the same constant
for $j \in \{1, \cdots, d_i\}$.

Let us further rewrite $f_{\theta} (x)$ 
in Eq.~\eqref{eq:k-deep-nn-form} as 
\begin{equation}
f_{\theta} (x) = 
\mathcal{G}_j^{(i)}
\left(
	\left( \sum_{v = 1}^{d_{i-1}}{w_{j, v}^{(i)} C_{v, 1}^{(i-1)} } \right)  x_1 + \cdots +
	\left( \sum_{v = 1}^{d_{i-1}}{w_{j, v}^{(i)} C_{v, d_0}^{(i-1)} } \right)  x_{d_0}
\right)
+ B_{\mathcal{P}}
\end{equation}
where $C_{v, u}^{(i-1)} \in \mathbb{R}$ is the $u$-th element of 
the $v$-th row vector of $C^{(i-1)}$. 


In Step 2, using the zero-deep neural network extraction attack,
we have recovered the following $d_i$ 
weight vectors ($j \in \{1, \cdots, d_i\}$)
\begin{equation}  \label{eq:k-nn-step-2}
\frac{
\varGamma_{\mathcal{P}}
}{
\left| \varGamma_{\mathcal{P}, 1} \right|
} =
\left[
\frac{
\mathcal{G}_j^{(i)}
\left( \sum_{v = 1}^{d_{i-1}}{w_{j, v}^{(i)} C_{v, 1}^{(i-1)} } \right)
}
{
\left|  
\mathcal{G}_j^{(i)}
\left( \sum_{v = 1}^{d_{i-1}}{w_{j, v}^{(i)} C_{v, 1}^{(i-1)} } \right)
\right|
},
\cdots, 
\frac{
\mathcal{G}_j^{(i)}
\left( \sum_{v = 1}^{d_{i-1}}{w_{j, v}^{(i)} C_{v, d_0}^{(i-1)} } \right)
}
{
\left|  
\mathcal{G}_j^{(i)}
\left( \sum_{v = 1}^{d_{i-1}}{w_{j, v}^{(i)} C_{v, 1}^{(i-1)} } \right)
\right|
}
\right].
\end{equation}
In this step, our target is to obtain the weight vector
$\widehat{A}_j^{(i)} = \left[ \widehat{w}_{j, 1}^{(i)}, 
\cdots, \widehat{w}_{j, d_{i-1}}^{(i)} \right]$.

It is clear that we need to guess the sign of 
$\mathcal{G}_j^{(i)}$
for $j \in \{1, \cdots, d_i\}$.
Again, all the $d_i$ signs are the same.
Consider the expression
in Eq.~\eqref{eq:k-deep-nn-form}.
Since the $j$-th neuron is active, its output exceeds $0$, i.e.,
\begin{equation}
A_j^{(i)}  C^{(i-1)}  x +
	B_{\left( \mathcal{P}^{(1)}, \cdots, \mathcal{P}^{(i)} \right)}  >  0, 
j \in \{1, \cdots, d_{i}\}.   \nonumber
\end{equation}
Then 
$\mathcal{G}_j^{(i)} \times
B_{\left( \mathcal{P}^{(i+1)}, \cdots, \mathcal{P}^{(k)} \right)} < 0 $
holds for $j \in \{1, \cdots, d_i\}$.
At the same time,
since 
$B_{\left( \mathcal{P}^{(i+1)}, \cdots, \mathcal{P}^{(k)} \right)}$ 
is a constant,
all the $d_i$ signs are the same.
Therefore, by guessing one sign, i.e., 
the sign of $\mathcal{G}_j^{(i)}$,
based on Eq.~\eqref{eq:k-nn-step-2},
we obtain
\begin{equation}  \label{eq:k-nn-step-2-v2}
\left[
\frac{
\sum_{v = 1}^{d_{i-1}}{w_{j, v}^{(i)} C_{v, 1}^{(i-1)} }
}
{
\left|  
\sum_{v = 1}^{d_{i-1}}{w_{j, v}^{(i)} C_{v, 1}^{(i-1)} } 
\right|
},
\cdots, 
\frac{
 \sum_{v = 1}^{d_{i-1}}{w_{j, v}^{(i)} C_{v, d_0}^{(i-1)} } 
}
{
\left|  
 \sum_{v = 1}^{d_{i-1}}{w_{j, v}^{(i)} C_{v, 1}^{(i-1)} }
\right|
}
\right], j \in \{1, \cdots, d_i\}.
\end{equation}

Note that $\widehat{C}_{v, u}^{(i-1)}$ can be obtained
using Eq.~\eqref{eq:current_G_C},
since $\widehat{A}^{(1)}, \cdots, \widehat{A}^{(i-1)}$ are known.
Then, basing the vector in Eq.~\eqref{eq:k-nn-step-2-v2},
we build a system of linear equations
\begin{equation}	\label{eq:soe-of-k-nn}
\left\{
\begin{array}{c}
\sum_{v = 1}^{d_{i-1}}{\widehat{w}_{j, v}^{(i)} \widehat{C}_{v, 1}^{(i-1)} }
=
\frac{
\sum_{v = 1}^{d_{i-1}}{w_{j, v}^{(i)} C_{v, 1}^{(i-1)} }
}
{
\left|  
\sum_{v = 1}^{d_{i-1}}{w_{j, v}^{(i)} C_{v, 1}^{(i-1)} } 
\right|
} , \\
\vdots   \\
\sum_{v = 1}^{d_{i-1}}{\widehat{w}_{j, v}^{(i)} \widehat{C}_{v, d_0}^{(i-1)} }
=
\frac{
\sum_{v = 1}^{d_{i-1}}{w_{j, v}^{(i)} C_{v, d_0}^{(i-1)} }
}
{
\left|  
\sum_{v = 1}^{d_{i-1}}{w_{j, v}^{(i)} C_{v, 1}^{(i-1)} } 
\right|
} , \\
\end{array}
\right.
\end{equation}
When $d_0 \geqslant d_{i-1}$
\footnote[7]{
The case of $d_0 \geqslant d_{i-1}$
is common in various applications, 
particularly in computer vision~\cite{DBLP:conf/ccs/GanjuWYGB18, 
DBLP:conf/wacv/LongCSH19, DBLP:conf/iccv/PerazziWGS15}, 
since the dimensions of images or videos are often large.}, 
we obtain $\widehat{A}_{j}^{(i)} =
\left[ \widehat{w}_{j,1}^{(i)}, \cdots,  \widehat{w}_{j, d_{i-1}}^{(i)}\right]$
by solving the above system of linear equations.
Lemma~\ref{lemma:w-i-of-k-nn} summarizes the expression of 
extracted weight vectors 
$\widehat{A}_j^{(i)}, j \in \{1, \cdots, d_i\}, i \in \{2, \cdots, k+1\}$.

\begin{lemma}	\label{lemma:w-i-of-k-nn}
Based on the system of linear equations 
presented in Eq.~\eqref{eq:soe-of-k-nn},
for $i \in \{2, \cdots, k+1\}$
and $j \in \{1, \cdots, d_i\}$,
the extracted weight vector 
$\widehat{A}_j^{(i)} = 
\left[ \widehat{w}_{j,1}^{(i)}, \cdots,  
\widehat{w}_{j, d_{i-1}}^{(i)}\right]$ is
\begin{equation}	\label{eq:expected_Ai}
\widehat{A}_{j}^{(i)} 
= \left[
\frac{
w_{j,1}^{(i)} \times \left| \sum_{v = 1}^{d_{i-2}}{w_{1, v}^{(i-1)} C_{v, 1}^{(i-2)} } \right|
}
{
\left| \sum_{v = 1}^{d_{i-1}}{w_{j, v}^{(i)} C_{v, 1}^{(i-1)} }   \right|
},
\cdots,
\frac{
w_{j,d_{i-1}}^{(i)} \times 
\left| \sum_{v = 1}^{d_{i-2}}{w_{d_{i-1}, v}^{(i-1)} C_{v, 1}^{(i-2)} } \right|
}
{
\left| \sum_{v = 1}^{d_{i-1}}{w_{j, v}^{(i)} C_{v, 1}^{(i-1)} }   \right|
}
\right],   
\end{equation}
where 
$ C_{v, 1}^{(q)} = A_{v}^{(q)} A^{(q-1)} \cdots A^{(2)} 
                \left[ A_{1, 1}^{(1)}, \cdots, A_{d_1, 1}^{(1)} \right]^\top $ .
\end{lemma}
\begin{proof}
The proof refers to Appendix~\ref{appendix:proof_lemma_1}.
\end{proof}

In Lemma~\ref{lemma:w-i-of-k-nn},
for the consistency of the mathematical symbols,
the weights $A^{(k+1)}$ 
are denoted by $[w_{1, 1}^{(k+1)}, \cdots, w_{1, d_{k}}^{(k+1)}]$
instead of $[w_{1}^{(k+1)}, \cdots, w_{d_{k}}^{(k+1)}]$.

\paragraph{\textup{\textbf{Recover All the Biases.  }}}
Since $\widehat{A}^{(i)}$ for 
$i \in \{1, \cdots, k+1\}$ have been obtained,
we can extract all the biases
by solving a system of linear equations.

Concretely, 
for the $\sum_{i=1}^{k}{d_i} + 1$ decision boundary points,
$f_{\widehat{\theta}}(x) = 0$ should hold.
Thus, we build a system of linear equations:
$f_{\widehat{\theta}}(x) = 0$ where
the expression of $f_{\widehat{\theta}}(x)$ refers to 
Eq.~\eqref{eq:k-nn-form}.
Combining with Lemma~\ref{lemma:w-i-of-k-nn},
by solving the above system, we will obtain
\begin{equation}	\label{eq:k-nn-biases}
\left\{
\begin{array}{l}
\widehat{b}^{(i)} =
\left[
\frac{
b_{1}^{(i)}
}
{
\left| \sum_{v = 1}^{d_{i-1}}{w_{1, v}^{(i)} C_{v, 1}^{(i-1)} }   \right|
},
\cdots,
\frac{
b_{d_{i}}^{(i)} 
}
{
\left| \sum_{v = 1}^{d_{i-1}}{w_{d_{i}, v}^{(i)} C_{v, 1}^{(i-1)} }   \right|
}
\right], i \in \{1, \cdots, k\}\\

\widehat{b}^{(k+1)} = \frac{b^{(k+1)}}
{
\left| \sum_{v = 1}^{d_{k}}{w_{v}^{(k+1)} C_{v, 1}^{(k)} }   \right|
}  .\\
\end{array}
\right. 
\end{equation}

Based on the extracted neural network parameters
(see Eq.~\eqref{eq:k-attack-step-2-target},
Eq.~\eqref{eq:expected_Ai},
and Eq.~\eqref{eq:k-nn-biases}),
the model signature of the extracted 
model $f_{\widehat{\theta}}$ is 
\begin{equation}
\mathcal{S}_{\widehat{\theta}} = 
\left \{
(\widehat{\varGamma}_{\mathcal{P}}, \widehat{B}_{\mathcal{P}}) = 
\left(
\frac{\varGamma_{\mathcal{P}}}{
\left| \sum_{v = 1}^{d_{k}}{w_{v}^{(k+1)} C_{v, 1}^{(k)} }   \right|
},
\frac{B_{\mathcal{P}}}{
\left| \sum_{v = 1}^{d_{k}}{w_{v}^{(k+1)} C_{v, 1}^{(k)} }   \right|
}
\right)
{\rm for} \,\, {\rm all} \,\, {\rm the } \,\, \mathcal{P}{\rm's}
\right \} .	\nonumber
\end{equation}

Consider the $j$-th neuron $\eta$ in layer $i$.
For an input $x \in \mathcal{X}$,
denote by $h(\eta; x)$
the output of the neuron of the victim model.
Based on the extracted neural network parameters
(see Eq.~\eqref{eq:k-attack-step-2-target}, 
Eq.~\eqref{eq:expected_Ai},
and Eq.~\eqref{eq:k-nn-biases}),
the output of the neuron of the extracted model
is $\frac{h(\eta; x)}{ \left| 
\sum_{v=1}^{d_{i-1}}{w_{j,v}^{(i)} C_{v, 1}^{(i-1)}} \right|}$.
At the same time, all the weights $w_{?,j}^{(i+1)}$ in layer $i+1$
are increased by a factor of 
$\left|  \sum_{v=1}^{d_{i-1}}{w_{j,v}^{(i)} C_{v, 1}^{(i-1)}} \right|$.
Thus, for the victim model and extracted model,
the $j$-th neuron in layer $i$ has the same influence
on all the neurons in layer $i+1$. 
As a result, for any $x \in \mathcal{X}$,
the model activation pattern of the victim model
is the same as that of the extracted model.
Combining with $\mathcal{S}_{\widehat{\theta}}$, 
for all $x \in \mathcal{X}$, we have
\begin{equation}
f_{\widehat{\theta}} (x) = 
\frac{1}{
\left| \sum_{v = 1}^{d_{k}}{w_{v}^{(k+1)} C_{v, 1}^{(k)} }   \right|
}  \times f_{\theta}(x) .
\end{equation}

In Appendix~\ref {appendix:attack_on_1_deep_nn}, 
we apply the extraction attack on $1$-deep neural networks 
and directly present the extracted model, 
which helps further understand our attack.

\begin{remark}	 \label{remark:change_MAPs}
Except for the $n+1$ model activation
patterns as shown in Eq.~\eqref{eq:MAPs_for_layer_1}
and Eq.~\eqref{eq:MAPs_for_layer_i},
the adversary could choose a new set of $n+1$
model activation patterns.
The reason is as follows.
Consider the recovery of the weight vector of
the $j$-th neuron in layer $i$,
and look at Fig.~\ref{fig:overview_of_attack} again.
Our attack only requires that:
(1) in layer $i$, only the $j$-th neuron is active;
(2) in layer $i-1$, all the $d_{i-1}$ neurons are active.
The neuron states in other layers
do not affect the attack.
Thus, there are more options for 
the $n+1$ model activation patterns,
and the rationale does not change.
\end{remark}

\paragraph{Discussion on The Computation Complexity. }
Once $n+1$ decision boundary points and 
$k$ sign guesses are selected,
to obtain an extracted model,
we just need to solve $n+2-d_1$ systems of linear equations.
However, since the model activation pattern
of a decision boundary point is unknown, 
we have to traverse all the possible 
combinations of $n+1$ decision boundary points
(see Step 3 in Section~\ref{sec:overview}),
which is the bottleneck of the total computation complexity.
The complete analysis of the attack complexity is presented in
Appendix~\ref{appendix:attack_complexity}.
\section{Instantiating the Extraction Attack in Practice}
\label{sec:refine_attacks}

Recall that, the complete extraction attack contains $5$ steps
introduced in Section~\ref{sec:overview}.
To obtain a functionally equivalent model,
the adversary also needs three auxiliary techniques: 
\emph{finding decision boundary points} (related to Steps 1 and 2), 
\emph{filtering duplicate affine transformations} (related to Step 2), 
and \emph{filtering functionally inequivalent models} (related to Step 5).

The idealized extraction attack 
introduced in Section~\ref{sec:DL_attack} 
relies on decision boundary points $x$ 
that make $f_{\theta} (x) = 0$ strictly hold.
This section will propose a binary searching method
to find decision boundary points under the hard-label setting.
Under finite precision,
it is hard to find decision boundary points $x$ 
that make $f_{\theta} (x) = 0$ strictly hold.
Therefore, the proposed method returns
input points $x$ close to the decision hyperplane
as decision boundary points.
As a result, the remaining two techniques need to
consider the influence of finite precision.  
This ensures our model extraction attacks work in practice, 
for producing a $(\varepsilon, 0)$-functionally equivalent model.

%

\subsection{Finding Decision Boundary Points}
\label{subsec:find_DBP}
Let us see how to find decision boundary points 
under the hard-label setting.
Fig.~\ref{fig:boundary_points} shows a schematic diagram
in a 2-dimensional input space.

\begin{figure}[!htb]
\centering
\includegraphics[width=0.3\textwidth]{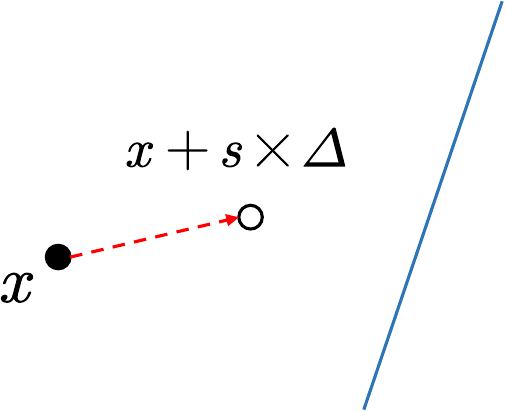}
\caption{
A schematic diagram of finding decision boundary points.
The blue solid line stands for the decision hyperplane
composed of decision boundary points.
The red dashed line stands for a direction vector 
$\varDelta \in \mathbb{R}^{d_0}$.
The starting point $x \in \mathbb{R}^{d_0}$ 
(i.e., the solid black circle) 
moves along the direction $\varDelta$,
and arrives at $x + s \times \varDelta$ 
(i.e., the hollow black circle)
where $s \in \mathbb{R}$ is the moving stride.
}
\label{fig:boundary_points}
\end{figure}

We first randomly pick a starting point $x \in \mathbb{R}^{d_0}$
and non-zero direction vector $\varDelta \in \mathbb{R}^{d_0}$.
Then let the starting point move along the direction $\varDelta$
or the opposite direction $- \varDelta$.
It is expected that the starting point 
will eventually cross the decision hyperplane in one direction,
as long as $\varDelta$ and $-\varDelta$ are not parallel
to the decision hyperplane.

Denote by $s \in \mathbb{R}$ 
the moving stride of the starting point,
which means that the starting point arrives at $x + s \times \varDelta$.
After querying the Oracle with $x$ and $x + s \times \varDelta$, 
if $z\left( f_{\theta} (x) \right) \ne 
z\left( f_{\theta} (x + s \times \varDelta) \right)$
(i.e., the two labels are different),
we know that the starting point has crossed the decision hyperplane 
when the moving stride is $s$.
Now, the core of finding decision boundary points is 
to determine a suitable moving stride $s$, 
such that the starting point reaches the decision hyperplane, 
i.e., $f_{\theta} (x + s \times \varDelta) = 0$ holds.

This task is done by \emph{binary search}.
Concretely, randomly choose two different 
moving strides $s_{\rm slow}$ and $s_{\rm fast}$ at first, such that 
\begin{equation}	\label{eq:find_boundary_points}
\begin{aligned}
z\left( f_{\theta} (x + s_{\rm slow} \times \varDelta) \right) &=
z\left( f_{\theta} (x) \right) ,\\
z\left( f_{\theta} (x + s_{\rm slow} \times \varDelta) \right) &\ne
z\left( f_{\theta} (x + s_{\rm fast} \times \varDelta) \right) .
\end{aligned}
\end{equation}
Then,
without changing the conditions presented 
in Eq.~\eqref{eq:find_boundary_points},
we dynamically adjust $s_{\rm slow}$ and $s_{\rm fast}$
until their absolute difference is close to $0$,
i.e., $\left| s_{\rm slow} - s_{\rm fast} \right| < \epsilon$
where $\epsilon$ is a precision defined by the adversary.
Finally, return $x + s_{\rm slow} \times \varDelta$ 
as a decision boundary point.

Since the precision $\epsilon$ is finite,
$x + s_{\rm slow} \times \varDelta$ 
is not strictly at the decision boundary,
which will inevitably introduce minor errors (equivalent to noises)
into the extracted model.
If $\epsilon$ decreases, 
then $x + s_{\rm slow} \times \varDelta$ will be closer to
the decision boundary,
which is helpful to the model extraction attack,
refers to the experiment results 
in Section~\ref{sec:experiments}.

\subsection{Filtering Duplicate Affine Transformations}
\label{subsec:filter_duplicate_AT}

For a $k$-deep neural network $f_{\theta}$ 
consisting of $n = \sum_{i=1}^{k}{d_i}$ neurons,
the idealized extraction attack 
exploits special $n+1$ model activation patterns.

To ensure that the required $n+1$ model activation patterns
occur with a probability as high as possible,
in Step 1 introduced in Section~\ref{sec:overview},
we collect $M$ decision boundary points where $M \gg n+1$,
e.g., $M = c_n 2^n$ and $c_n$ is a small factor.
As a result, there are many collected decision boundary points 
with duplicate model activation patterns. 
Therefore, in Step 2, after recovering the parameter tuple
$\left( \varGamma_{\mathcal{P}}, B_{\mathcal{P}} \right)$
(i.e., the affine transformation)
corresponding to each decision boundary point,
we need to filter the decision boundary points 
with duplicate affine transformations,
since their model activation patterns should be the same. 
When filtering duplicate affine transformations, 
we consider two possible cases.

\paragraph{\textup{\textbf{Filtering Correctly Recovered Affine Transformations. }}}
In the first case, 
assume that two affine transformations are both correctly recovered.

However, recovering affine transformations 
(i.e., 0-deep neural network extraction attack)
relies on finding decision boundary points,
which introduces minor errors.
This is equivalent to adding noises to
the recovered affine transformations, i.e.,
the tuples $\left( \varGamma_{\mathcal{P}}, 
B_{\mathcal{P}} \right)$.
To check whether two noisy affine transformations are the same,
we adopt the checking rule below.

\paragraph{Comparing two vectors. }
Consider two vectors with the same dimension,
e.g., $V^{1} \in \mathbb{R}^{d}, V^{2} \in \mathbb{R}^{d}$.
Set a small threshold $\varphi$.
If the following $d$ inequations hold simultaneously
\begin{equation}	\label{eq:compare_rule}
\left| V_j^{1} - V_j^{2} \right| < \varphi, \,\, j \in \{1, \cdots d\}
\end{equation}
where $V_j^{1}$ and $V_j^{2}$ are, respectively, 
the $j$-th element of $V^{1}$ and $V^{2}$,
the two vectors are considered to be the same.


\paragraph{\textup{\textbf{
Filtering Wrongly Recovered Affine Transformations. }}}

In the second case,
assume that
one affine transformation is correctly recovered
and another one is partially recovered.

For the extraction attack on $k$-deep neural networks,
when recovering the affine transformation 
corresponding to an input by the 
$0$-deep neural network extraction attack 
(see Section~\ref{subsec:0_nn_attack}),
the process of binary search 
should not change the model activation pattern.
Otherwise, the affine transformation may be wrongly recovered.
Recall that, in the $0$-deep neural network extraction attack, 
the $d_0$ elements of $\varGamma_{\mathcal{P}}$ 
are recovered one by one independently.
Thus, the wrong recovery of one element of 
$\varGamma_{\mathcal{P}}$ 
does not influence the recovery of other elements.

As a result, we have to consider the case that 
one transformation is partially recovered.
In this case, the filtering method is as follows.
Consider two vectors $V^{1} \in \mathbb{R}^{d}$
and $V^{2} \in \mathbb{R}^{d}$.
If $\left| V_j^{1} - V_j^{2} \right| < \varphi$ 
holds for at least $(d - d_{\varphi})$ $j$'s where $j \in \{1, \cdots, d\}$
and $d_{\varphi} \in \mathbb{N}$ is a threshold,
the two vectors are considered to be the same.
Suppose that the occurrence frequencies of 
$V^{1}$ and $V^{2}$ are $o_1$ and $o_2$ respectively,
we regard $V^{1}$ as the correctly recovered affine transformation 
if $o_1 \gg o_2$, and vice versa.

\subsection{Filtering Functionally Inequivalent Extracted Models}
\label{subsec:filter_wrong_models}

Consider $k$-deep neural networks consisting of 
$n = \sum_{i=1}^{k}{d_i}$ neurons.
As introduced in Section~\ref{sec:overview},
each time we randomly choose $n+1$ out of $N$ collected decision
boundary points to generate an extracted model.
Moreover, 
according to Section~\ref{subsec:k_deep_nn_attack},
in the extraction attack,
we need to guess $k$ signs, i.e., 
the sign of $\mathcal{G}_j^{(i)}, i \in \{1, \cdots, k\}$.

When the model activation patterns of the selected 
$n+1$ decision boundary points 
are not those required in the extraction attack,
or at least one of the $k$ sign guesses is wrong,
the resulting extracted model $f_{\widehat{\theta}}$ 
is not a functionally equivalent model of 
the victim model $f_{\theta}$.
Thus, we will get many functionally inequivalent extracted models.

Besides, due to the minor errors introduced 
by the finite precision used in finding decision boundary points,
the parameters of the extracted model may be slightly different from
the theoretical values (see Eq.~\eqref{eq:k-attack-step-2-target}, 
Eq.~\eqref{eq:expected_Ai}, and Eq.~\eqref{eq:k-nn-biases}).
This subsection introduces three methods to 
filter functionally inequivalent extracted models, 
one of which considers the negative influence of finite precision together.


\paragraph{\textup{\textbf{Filtering
by the Normalized Model Signature.  }}}
Before introducing the filtering method,
we discuss how many possible model activation 
patterns there are at most for a $k$-deep neural network.
Lemma~\ref{lem:num_of_MAPs} answers this question.

\begin{lemma}	\label{lem:num_of_MAPs}
For a $k$-deep neural network consisting of 
$n = \sum_{i=1}^{k}{d_i}$ neurons,
the upper bound of the number of possible model activation patterns is
\begin{equation}
H = 
\left( \prod_{i=1}^{k}{(2^{d_i} - 1)} \right)  +  
\sum_{i=2}^{k}{\left(  \prod_{j=1}^{i-1}{(2^{d_j} - 1)} \right)} ,
\end{equation}
where $d_i$ is the number of neurons in layer $i$.
\end{lemma}
\begin{proof}
If all the $d_i$ neurons in layer $i$ are inactive,
i.e., the outputs of these neurons are $0$,
then the neuron states of all the $\sum_{j=i+1}^{k}{d_j}$ neurons
in the last $k - i$ layers are deterministic. 
In this case, the number of possible model activation patterns
is decided by the first $i-1$ layers, i.e.,
the maximum is $ \prod_{j=1}^{i-1}{(2^{d_j} - 1)}$.
If there is at least one active neuron in each layer,
then there are at most $\prod_{i=1}^{k}{(2^{d_i} - 1)}$ possible
model activation patterns.
\end{proof}

After all the weights $\widehat{A}^{(i)}$ 
and biases $\widehat{b}^{(i)}, i \in \{1, \cdots, k+1\}$ are obtained,
we assume that 
all the $H$ model activation patterns are possible,
and compute the resulting normalized model signature 
$\mathcal{S}_{\widehat{\theta}}^{\mathcal{N}}$. 
Denote by $\mathcal{S}_{\theta}^{\mathcal{N}}$
the normalized model signature recovered in Step 2 
(see Section~\ref{sec:overview}).
If $\mathcal{S}_{\theta}^{\mathcal{N}}$ is not a subset of
$\mathcal{S}_{\widehat{\theta}}^{\mathcal{N}}$,
we regard $f_{\widehat{\theta}}$ as a
functionally inequivalent model.


Due to the minor errors caused by finite precision,
i.e.,  the slight difference between the extracted parameters 
$\widehat{\theta}$ and the theoretical values 
(see Eq.~\eqref{eq:k-attack-step-2-target}, 
Eq.~\eqref{eq:expected_Ai}, and Eq.~\eqref{eq:k-nn-biases}), 
when checking whether a tuple 
$\left( \varGamma_{\mathcal{P}}, 
B_{\mathcal{P}} \right) 
\in \mathcal{S}_{\theta}^{\mathcal{N}}$ 
is equal to a tuple
$\left( \widehat{\varGamma}_{\mathcal{P}}, 
\widehat{B}_{\mathcal{P}} \right) 
\in \mathcal{S}_{\widehat{\theta}}^{\mathcal{N}}$ or not,
we adopt the checking rule presented in
Section~\ref{subsec:filter_duplicate_AT},
refers to Eq.~\eqref{eq:compare_rule}.


Besides, 
the filtering method in Section~\ref{subsec:filter_duplicate_AT}
does not ensure that all the wrongly recovered 
affine transformations are filtered. 
To avoid the functionally equivalent model being filtered,
we adopt a flexible method.

Recall that, in Step 1, 
we collect a sufficient number of decision boundary points.
For each tuple $\left( \varGamma_{\mathcal{P}}, 
B_{\mathcal{P}} \right) 
\in \mathcal{S}_{\theta}^{\mathcal{N}}$,
denote by $m_{\mathcal{P}}$ the frequency that
the tuple occurs in the collected decision boundary points.
Suppose that the number of $m_{\mathcal{P}}$ 
where $m_{\mathcal{P}} > 1$ is $N_{\rm valid}$.
Then when at least $0.95 \times N_{\rm valid}$
tuples $\left( \varGamma_{\mathcal{P}}, 
B_{\mathcal{P}} \right) 
\in \mathcal{S}_{\theta}^{\mathcal{N}}$ are
in the set $\mathcal{S}_{\widehat{\theta}}^{\mathcal{N}}$, 
the extracted model $f_{\widehat{\theta}}$ 
is regarded as a candidate of the functionally equivalent model.
Here, We call the ratio $0.95 \times \frac{N_{\rm valid}}{
\left|  \mathcal{S}_{\theta}^{\mathcal{N}} \right|}$
the adaptive threshold.

\paragraph{\textup{\textbf{Filtering by Weight Signs.  }}}
After $\widehat{A}^{(i)}$, 
$\widehat{b}^{(i)}$ for $i \in \{1, \cdots, k+1\}$ are obtained,
we compute the matrices $\widehat{\mathcal{G}}^{(i)}$
and check whether the $k$ signs, i.e., the sign 
of $\widehat{\mathcal{G}}_j^{(i)}, i \in \{1, \cdots, k\}$
are consistent with the $k$ guesses.
If at least one sign is not consistent with the guess, 
the extracted model is not the functionally equivalent model. 

Interestingly, except for handling wrong sign guesses, 
this method also shows high filtering effectiveness 
when the model activation patterns of 
the selected $n + 1$ decision boundary points 
are not those required by our extraction attacks. 
This is not strange, 
since our extraction attack is designed 
for a specific set of model activation patterns.   
For wrong model activation patterns, 
whether the sign of 
$\widehat{\mathcal{G}}_j^{(i)}, i \in \{1, \cdots, k\}$ is $1$ or $-1$ 
is a random event.

\paragraph{\textup{\textbf{Filtering by Prediction Matching Ratio. }}}
The above two filtering methods are effective,
but we find that some functionally inequivalent 
models still escape from the filtering.
Therefore, the third method is designed to
perform the last filtering on extracted models surviving from
the above two filtering methods.
This method is based on the \emph{prediction matching ratio}.

\paragraph{Prediction Matching Ratio. }
Randomly generate $N_1$ inputs,
query the extracted model $f_{\widehat{\theta}}$ 
and the victim model $f_{\theta}$.
Suppose that the two models return the same hard-label for 
$N_2$ out of $N_1$ inputs.
The ratio $\frac{N_2}{N_1}$ is called the prediction matching ratio.

According to Definition~\ref{def:EFEE} and 
Definition~\ref{def:equivalence-in-practice},
for a functionally equivalent model, 
the prediction matching ratio should be high, 
or even close to $100\%$.
Note that many random inputs $x$ 
and corresponding hard-label $z\left( f_{\theta} (x) \right)$
are collected during the attack process 
(see Steps 1 and 2 in Section~\ref{sec:overview}).
Thus, we can exploit these inputs.


\section{Experiments}
\label{sec:experiments}

Our model extraction attacks are evaluated on 
both untrained and trained neural networks.
Concretely, we first perform experiments on 
untrained neural networks with diverse architectures 
and randomly generated parameters. 
Then, based on two typical benchmarking image datasets 
(i.e., MNIST, CIFAR10)
in visual deep learning, 
we train a series of neural networks as classifiers 
and evaluate the model extraction attacks 
on these trained neural networks.

For convenience, denote by `$d_0$-$d_1$-$\cdots$-$d_{k+1}$'
the victim model, where $d_i$ is the dimension of each layer.
For example, the symbol 1000-1 stands for a $0$-deep neural network
with an input dimension of $1000$ and an output dimension of $1$.

\paragraph{Partial Universal Experiment Settings. }
Some settings are used in all the following experiments.
For $k$-deep neural network extraction attacks, in Step 1,
we randomly generate $8 \times 2^n$ pairs of 
starting point and moving direction, 
where $n = \sum_{i=1}^{k}{d_i}$ is the number of neurons.
The prediction matching ratio is estimated
over $10^6$ random inputs.

\subsection{Computing $(\varepsilon, 0)$-Functional Equivalence}

To quantify the degree
to which a model extraction attack has succeeded,
the method (i.e., error bounds 
propagation~\cite{DBLP:conf/crypto/CarliniJM20}) 
proposed by Carlini et al. is adopted to 
compute $(\varepsilon, 0)$-functional equivalence.

\paragraph{\textup{\textbf{Error bounds propagation. }}}
To compute $(\varepsilon, 0)$-functional equivalence of 
the extracted neural network $f_{\widehat{\theta}}$,
one just needs to compare the extracted parameters 
(weights $\widehat{A}^{(i)}$ and biases $\widehat{b}^{(i)}$) 
to the real parameters 
(weights $A^{(i)}$ and biases $b^{(i)}$)
and analytically derive an upper bound on the error 
when performing inference~\cite{DBLP:conf/crypto/CarliniJM20}. 

Before comparing the neural network parameters,
one must `align' them~\cite{DBLP:conf/crypto/CarliniJM20}. 
This involves two operations:
(1) adjusting the order of the neurons in the network,
i.e., the order of the rows or columns of $A^{(i)}$ and $b^{(i)}$,
(2) adjusting the values of $A^{(i)}$ and $b^{(i)}$ to 
the theoretical one (see Eq.~\eqref{eq:k-attack-step-2-target},
Eq.~\eqref{eq:expected_Ai}, and Eq.~\eqref{eq:k-nn-biases})
obtained by the idealized model extraction attacks.
This gives an aligned $\widetilde{A}^{(i)}$ and $\widetilde{b}^{(i)}$
from which one can analytically derive upper bounds on the error.
Other details (e.g., propagating error bounds layer-by-layer) are the same
as that introduced in~\cite{DBLP:conf/crypto/CarliniJM20},
and not introduced again in this paper.

\subsection{Experiments on Untrained Neural Networks}
Table~\ref{tab:attack_on_1_deep_nn} summarizes 
the experimental results on different untrained neural networks
which demonstrates the effectiveness of our model extraction attacks.

\begin{table}[!htbp]
\centering
\renewcommand\arraystretch{1.0}
\caption{Experiment results on untrained $k$-deep neural networks.}
\label{tab:attack_on_1_deep_nn}
\begin{tabular}{c|c|c|c|c|c|c}
\hline
Architecture	&	Parameters	&	$\epsilon$ 		&	
PMR		&	Queries		&	$(\varepsilon, 0)$	&	
max$|\theta - \widehat{\theta}|$	\\
\hline
512-2-1		&	1029			&	$10^{-12}$	&	
$100\%$		&	$2^{19.35}$	&	$2^{-12.21}$	&	
$2^{-16.88}$	\\
			&				&	$10^{-14}$ 	& 
$100\%$		&	$2^{19.59}$	&	$2^{-19.84}$	&	
$2^{-24.62}$	\\
\hline
2048-4-1		&	8201			&	$10^{-12}$	&	
$99.98\%$		&	$2^{23.32}$	&	$2^{-3.77}$	&	
$2^{-10.44}$	\\
			&				&	$10^{-14}$ 	& 
$100\%$		&	$2^{23.51}$	&	$2^{-13.70}$	&	
$2^{-17.75}$	\\
\hline
25120-4-1		&	100489		&	$10^{-14}$	&	
$99.98\%$		&	$2^{26.42}$	&	$2^{-2.99}$	&	
$2^{-14.67}$	\\
			&				&	$10^{-16}$ 	& 
$100\%$		&	$2^{26.67}$	&	$2^{-13.01}$	&	
$2^{-23.19}$	\\
\hline
50240-2-1		&	100485		&	$10^{-14}$	&	
$99.99\%$		&	$2^{25.85}$	&	$2^{-7.20}$	&	
$2^{-15.58}$	\\
			&				&	$10^{-16}$ 	& 
$100\%$		&	$2^{26.31}$	&	$2^{-14.44}$	&	
$2^{-22.67}$	\\
\hline
32-2-2-1		&	75			&	$10^{-12}$	&	
$100\%$		&	$2^{17.32}$	&	$2^{-10.99}$	&	
$2^{-14.78}$	\\
			&				&	$10^{-14}$	&	
$100\%$		&	$2^{17.56}$	&	$2^{-18.21}$	&	
$2^{-20.61}$	\\
\hline
512-2-2-1		&	1035			&	$10^{-12}$	&	
$99.99\%$		&	$2^{21.39}$	&	$2^{-10.34}$	&	
$2^{-14.01}$	\\
			&				&	$10^{-14}$	&	
$100\%$		&	$2^{21.59}$	&	$2^{-14.17}$	&	
$2^{-17.29}$	\\
\hline
1024-2-2-1	&	2059			&	$10^{-12}$	&	
$99.99\%$		&	$2^{22.38}$	&	$2^{-6.10}$	&	
$2^{-13.77}$	\\
			&				&	$10^{-14}$	&	
$100\%$		&	$2^{22.49}$	&	$2^{-14.16}$	&	
$2^{-20.38}$	\\
\hline
\multicolumn{6}{l}{
$\epsilon$: the precision used to find decision boundary points. 
}		\\
\multicolumn{6}{l}{
max$|\theta - \widehat{\theta}|$: the maximum extraction error
of model parameters. 
}		\\
\multicolumn{6}{l}{
PMR: prediction matching ratio.
}		\\
\end{tabular}
\end{table}

According to Appendix~\ref{appendix:attack_complexity},
the computation complexity of our model extraction attack is
about $\mathcal{O}\left( n \times 2^{n^2+n+k} \right)$,
where $n$ is the number of neurons.
Thus, we limit the number of neurons, 
which does not influence the verification of 
our model extraction attack.
Note that the number of parameters is not limited.
All the attacks can be finished within several hours on a single core. 

The results in Table~\ref{tab:attack_on_1_deep_nn}
also support our argument in Remark~\ref{remark:change_MAPs}.
For the $2$-deep neural networks (e.g., 32-2-2-1),
when recovering the weights in layer $1$,
we require that only one neuron in layer $2$ is active,
instead of all the $2$ neurons being active.
Our extraction attacks also achieve good performance.

\paragraph{\textup{\textbf{The influence of the precision $\epsilon$.  }}}

A smaller $\epsilon$ will make the returned point 
$x + s_{\rm slow} \times \varDelta$ (see Section~\ref{subsec:find_DBP}) 
closer to the decision boundary,
which helps reduce the extraction error of affine transformations.
As a result, 
the model extraction attack is expected to perform better.
For example, for the 1-deep neural network 2048-4-1, 
when $\epsilon$ decreases from $10^{-12}$ to $10^{-14}$,
the value $\varepsilon$ (respectively, max$| \theta - \widehat{\theta} |$) 
decreases from $2^{-3.77}$ to $2^{-13.70}$ 
(respectively, from $2^{-10.44}$ to $2^{-17.75}$), 
which is a significant improvement.

At the same time,
using a smaller precision $\epsilon$ 
does not increase the attack complexity significantly.
According to Appendix~\ref{appendix:attack_complexity},
the query complexity is about
$\mathcal{O} \left( d_0 \times 2^n 
\times {\rm log}_2^{ \frac{1}{\epsilon} } \right)$.
Thus, decreasing $\epsilon$ has little influence on the query complexity.
Look at the neural network 2048-4-1 again.
When $\epsilon$ decreases from $10^{-12}$ to $10^{-14}$,
the number of queries only increases from $2^{23.32}$ to $2^{23.51}$.
Besides, when $n$ (i.e., the number of neurons) is large, 
$\epsilon$ almost does not 
influence the computation complexity,
since $\epsilon$ only influences Steps 1 and 2 
(see Section~\ref{sec:overview}),
while the computation complexity 
is mainly determined by other steps
(refer to Appendix~\ref{appendix:attack_complexity}).
When $n$ is small,
the practical runtime is determined by the query complexity,
then decreasing $\epsilon$ also has little influence
on the runtime.

Choosing an appropriate $\epsilon$ is simple.
In our experiments, 
we find that a smaller $\epsilon$ should be used,
when the prediction matching ratio 
estimated over $10^6$ random inputs is not $100\%$,
and the gap (e.g., $0.02\%$, see the third or fifth row) 
is not negligible.

\subsection{Experiments on Trained Neural Networks}


%

\paragraph{\textup{\textbf{The MNIST and CIFAR10 Dataset. }}}
MNIST (respectively, CIFAR10) is one typical 
benchmarking dataset used in visual deep learning.
It contains ten-class handwriting number 
gray images~\cite{DBLP:journals/pieee/LeCunBBH98}
(resp., real object images 
in a realistic environment~\cite{DBLP:journals/pami/TorralbaFF08}). 
Each of the ten classes, i.e., `0', `1', `2', `3', `4', `5', `6', `7', `8', and `9'
(resp., airplane, automobile, bird, cat, deer, dog, frog, horse, ship, truck),
contains $28 \times 28$ pixel gray images 
(resp., $32 \times 32$ pixel RGB images), 
totaling $60000$ (resp., $50000$) training and 
$10000$ (resp. $10000$) testing images.

\paragraph{\textup{\textbf{Neural Network Training Pipelines. }}}
When classifying different classes of objects,
the decision boundary of trained neural networks will be different.
To fully verify our model extraction attack,
for MNIST (respectively, CIFAR10), 
we divide the ten classes into five groups
and build a binary classification neural network for each group.
All the neural networks share the same architecture $d_0$-2-1,
where $d_0 = 28 \times 28$ for MNIST 
(respectively, $32 \times 32 \times 3$ for CIFAR10).
On the MNIST and CIFAR10 datasets,
we perform a standard rescaling of the pixel values from
$0 \cdots 255$ to $0 \cdots 1$.
For the model training,
we choose typical settings 
(the loss is the cross-entropy loss; 
the optimizer is standard stochastic gradient descent; 
batch size $128$).
The first four columns of 
Table~\ref{tab:attack_on_trained_nn} summarize a detailed 
description of the neural networks to be attacked in this section.

\paragraph{\textup{\textbf{Experiment Results.  }}}
The last four columns of 
Table~\ref{tab:attack_on_trained_nn} 
summarize the experiment results.
Our extraction attack still achieves good performance when
an appropriate precision $\epsilon$ is used,
which further verifies its effectiveness.

\begin{table}[!t]
\centering
\renewcommand\arraystretch{1.0}
\caption{Experiment results on neural networks 
trained on MNIST or CIFAR10. }
\label{tab:attack_on_trained_nn}
\begin{tabular}{c|c|c|c|c|c|c|c}
\hline
task		&	architecture	&	accuracy		&	parameters	&
$\epsilon$		&	Queries		&	$(\varepsilon, 0)$	&	
max$|\theta - \widehat{\theta}|$	\\
\hline
`0' vs `1'		&	784-2-1		&	0.9035		&	1573			&
$10^{-12}$	&	$2^{20.11}$	&	$2^{-16.39}$	&	$2^{-17.85}$	\\
			&				&				&				&
$10^{-14}$	&	$2^{20.32}$	&	$2^{-20.56}$	&	$2^{-22.81}$	\\
\hline
`2' vs `3'		&	784-2-1		&	0.8497		&	1573			&
$10^{-12}$	&	$2^{20.11}$	&	$2^{-7.00}$	&	$2^{-7.80}$	\\
			&				&				&				&
$10^{-14}$	&	$2^{20.32}$	&	$2^{-14.32}$	&	$2^{-15.06}$	\\
\hline
`4' vs `5'		&	784-2-1		&	0.8570		&	1573			&
$10^{-12}$	&	$2^{20.02}$	&	$2^{-8.47}$	&	$2^{-8.82}$ 	\\
			&				&				&				&
$10^{-14}$	&	$2^{20.32}$	&	$2^{-15.62}$	&	$2^{-15.81}$ 	\\
\hline
`6' vs `7'		&	784-2-1		&	0.9290		&	1573			&
$10^{-12}$	&	$2^{20.11}$	&	$2^{-7.02}$	&	$2^{-7.93}$ 	\\
			&				&				&				&
$10^{-14}$	&	$2^{20.32}$	&	$2^{-12.00}$	&	$2^{-12.91}$ 	\\
\hline
`8' vs `9'		&	784-2-1		&	0.9501		&	1573			&
$10^{-12}$	&	$2^{20.11}$	&	$2^{-10.58}$	&	$2^{-11.62}$ 	\\
			&				&				&				&
$10^{-14}$	&	$2^{20.32}$	&	$2^{-19.63}$	&	$2^{-21.72}$ 	\\
\hline
airplane vs 	&	3072-2-1		&	0.8120		&
6149			&	$10^{-12}$	&	$2^{22.08}$	&	$2^{-4.84}$	&	
$2^{-7.48}$ 	\\
automobile	&				&				&
			&	$10^{-14}$	&	$2^{22.29}$	&	$2^{-12.41}$	&	
$2^{-15.20}$ 	\\
\hline
bird vs cat		&	3072-2-1		&	0.6890		&	6149			&	
$10^{-12}$	&	$2^{22.07}$	&	$2^{-8.37}$	&	$2^{-9.80}$ 	\\
			&				&				&				&
$10^{-14}$	&	$2^{22.29}$	&	$2^{-12.27}$	&	$2^{-14.73}$ 	\\
\hline
deer vs dog	&	3072-2-1		&	0.6870		&	6149			&
$10^{-12}$	&	$2^{22.01}$	&	$2^{-9.55}$	&	$2^{-13.25}$ 	\\
			&				&				&				&
$10^{-14}$	&	$2^{22.22}$	&	$2^{-13.19}$	&	$2^{-15.82}$ 	\\
\hline
frog vs horse	&	3072-2-1		&	0.8405			&	6149			&
$10^{-12}$	&	$2^{22.08}$	&	$2^{-9.56}$	&	$2^{-10.71}$ 	\\
			&				&				&				&
$10^{-14}$	&	$2^{22.29}$	&	$2^{-13.58}$	&	$2^{-15.58}$ 	\\
\hline
ship vs truck	&	3072-2-1		&	0.7995		&	6149			&
$10^{-12}$	&	$2^{22.08}$	&	$2^{-8.63}$	&	$2^{-8.90}$ 	\\
			&				&				&				&
$10^{-14}$	&	$2^{22.29}$	&	$2^{-12.95}$	&	$2^{-13.02}$ 	\\
\hline
\multicolumn{6}{l}{
max$|\theta - \widehat{\theta}|$: the maximum extraction error
of model parameters. 
}		\\
\multicolumn{6}{l}{
accuracy: classification accuracy of the victim model $f_{\theta}$. 
}		\\
\multicolumn{6}{l}{
for saving space, prediction matching ratios are not listed. 
}		\\
\end{tabular}
\end{table}

The experimental results presented in 
Table~\ref{tab:attack_on_1_deep_nn} and 
Table~\ref{tab:attack_on_trained_nn} 
show that the attack performance 
(i.e., the value of $\varepsilon$ and 
$\text{max}|\theta - \widehat{\theta}|$)
 is related to the precision $\epsilon$ 
and the properties of the decision boundary. 
However, we do not find a clear quantitative relationship 
between the attack performance and the precision 
$\epsilon$ (or some unknown properties of the decision boundary). 
Considering that the unknown quantitative relationships 
do not influence the verification of the model extraction attack, 
we leave the problem of exploring 
the unknown relationships as a future work.

\section{Conclusion}
\label{sec:conclusion}

In this paper,  we have studied the model extraction attack 
against neural network models under the hard-label setting, 
i.e., the adversary only has access to the most likely class label 
corresponding to the raw output of neural network models. 
We propose new model extraction attacks that theoretically 
achieve functionally equivalent extraction.   
Practical experiments on numerous neural network models 
have verified the effectiveness of the proposed model extraction attacks.   
To the best of our knowledge, 
this is the first time to prove with practical experiments that 
it is possible to achieve functionally equivalent extraction 
against neural network models under the hard-label setting.

The future work will mainly focus on the following aspects:
\begin{itemize}
\item 
The (computation and query) complexity of 
our model extraction attack remains high,
which limits the application to neural networks 
with a large number of neurons.
Reducing the complexity is an important problem.

\item
In this paper, to recover the weight vector 
of the $j$-th neuron in layer $i$,
we require that in layer $i$, only the $j$-th neuron is active.
However, such a model activation pattern may not occur in some cases.
Then how to recover the weight vector of this neuron 
based on other model activation patterns 
would be a vital step towards better generality.

\item 
Explore possible quantitative relationships 
between the precision $\epsilon$ 
(or some unknown properties of the decision boundary) 
and $\varepsilon$ (or $\text{max} | \theta - \widehat{\theta} |$).

\item
Extend the extraction attack to the case of vector outputs, 
i.e., the output dimensionality exceeds $1$.

\item
Develop extraction attacks against
other kinds of neural network models.
\end{itemize}


\subsubsection{Acknowledgments. }
We would like to thank Adi Shamir for his guidance.
We would like to thank the anonymous reviewers for their detailed
and helpful comments.
This work was supported by the National Key R\&D Program of China (2018YFA0704701, 2020YFA0309705), 
Shandong Key Research and Development Program (2020ZLYS09), 
the Major Scientific and Technological Innovation Project of Shandong, 
China (2019JZZY010133), the Major Program of Guangdong Basic and Applied Research (2019B030302008), the Tsinghua University Dushi Program,
and the Ministry of Education in Singapore under Grant RG93/23.
Y. Chen was also supported by the Shuimu Tsinghua Scholar Program.

\appendix




\section{Proof of Lemma~\ref{lemma:w-i-of-k-nn}}
\label{appendix:proof_lemma_1}

We prove Lemma~\ref{lemma:w-i-of-k-nn} 
by Mathematical Induction.

\begin{proof}
When $i = 2$, according to Lemma~\ref{lemma:w-i-of-k-nn},
the extracted weight vector 
$\widehat{A}_j^{(2)}, j \in \{1, \cdots, d_2\}$ should be
\begin{equation}	\label{eq:expected_A2}
\begin{aligned}
\widehat{A}_{j}^{(2)} 
&= \left[
\frac{
w_{j,1}^{(2)} \times \left| \sum_{v = 1}^{d_{0}}{w_{1, v}^{(1)} C_{v, 1}^{(0)} } \right|
}
{
\left| \sum_{v = 1}^{d_{1}}{w_{j, v}^{(2)} C_{v, 1}^{(1)} }   \right|
},
\cdots,
\frac{
w_{j,d_{1}}^{(2)} \times 
\left| \sum_{v = 1}^{d_{0}}{w_{d_{1}, v}^{(1)} C_{v, 1}^{(0)} } \right|
}
{
\left| \sum_{v = 1}^{d_{1}}{w_{j, v}^{(2)} C_{v, 1}^{(1)} }   \right|
}
\right] \\
&=  \left[
\frac{
w_{j,1}^{(2)} \times \left| w_{1, 1}^{(1)} \right|
}
{
\left| \sum_{v = 1}^{d_{1}}{w_{j, v}^{(2)} C_{v, 1}^{(1)} }   \right|
},
\cdots,
\frac{
w_{j,d_{1}}^{(2)} \times 
\left| w_{d_{1}, 1}^{(1)} \right|
}
{
\left| \sum_{v = 1}^{d_{1}}{w_{j, v}^{(2)} C_{v, 1}^{(1)} }   \right|
}
\right]. \\
\end{aligned}
\end{equation}
Note that $C_{1, 1}^{(0)} = 1$ and
$C_{v, 1}^{(0)} = 0$ for $v \in \{2, \cdots, d_0\}$ .

Besides, we have
\begin{equation}
\begin{aligned}
C^{(1)} &= I_{\mathcal{P}}^{(1)} A^{(1)} = A^{(1)}
               = \left[ A_1^{(1)}, \cdots, A_{d_0}^{(1)} \right],  \\
\widehat{C}^{(1)} &= I_{\mathcal{P}}^{(1)} \widehat{A}^{(1)} = 
                                \widehat{A}^{(1)} = 
		\left[ \frac{A_1^{(1)}}{\left| w_{1,1}^{(1)} \right|}, \cdots, 
                           \frac{A_{d_0}^{(1)}}{\left| w_{d_0, 1}^{(1)} \right|} \right].  \\
\end{aligned}
\end{equation}
where $A_v^{(1)} = \left[ w_{v,1}^{(1)}, \cdots, w_{v, d_0}^{(1)} \right]$,
$C_{v, u}^{(1)} = w_{v,u}^{(1)}$ and 
$\widehat{C}_{v,u}^{(1)} = \frac{w_{v,u}^{(1)}}{\left| w_{v,1}^{(1)} \right| }$.

Look at the system of linear equations presented in
Eq.~\eqref{eq:soe-of-k-nn}.
Now, the system of linear equations 
is transformed into
\begin{equation}	
\left\{
\begin{array}{c}
\sum_{v = 1}^{d_{1}}{\widehat{w}_{j, v}^{(2)} \widehat{C}_{v, 1}^{(1)} }
=
\frac{
\sum_{v = 1}^{d_{1}}{w_{j, v}^{(2)} C_{v, 1}^{(1)} }
}
{
\left|  
\sum_{v = 1}^{d_{1}}{w_{j, v}^{(2)} C_{v, 1}^{(1)} } 
\right|
} \\
\vdots   \\
\sum_{v = 1}^{d_{1}}{\widehat{w}_{j, v}^{(2)} \widehat{C}_{v, d_0}^{(1)} }
=
\frac{
\sum_{v = 1}^{d_{1}}{w_{j, v}^{(2)} C_{v, d_0}^{(1)} }
}
{
\left|  
\sum_{v = 1}^{d_{1}}{w_{j, v}^{(2)} C_{v, 1}^{(1)} } 
\right|
} \\
\end{array}
\right.
\end{equation}
when $d_0 \geqslant d_1$, by solving the system, 
it is expected to obtain
\begin{equation}
\widehat{A}_j^{(2)} = \left[ 
\frac{
w_{j,1}^{(2)} \left| w_{1,1}^{(1)} \right|
}
{
\left|  
\sum_{v = 1}^{d_{1}}{w_{j, v}^{(2)} C_{v, 1}^{(1)} } 
\right|
},  \cdots, 
\frac{
w_{j,d_1}^{(2)} \left| w_{d_1,1}^{(1)} \right|
}
{
\left|  
\sum_{v = 1}^{d_{1}}{w_{j, v}^{(2)} C_{v, 1}^{(1)} } 
\right|
}, 
 \right],
\end{equation}
which is consistent with the expected value 
in Eq.~\eqref{eq:expected_A2}.

Next, consider the recovery of the weight vector of
the $j$-th neuron in layer $i$,
and assume that the weights $\widehat{A}^{(1)}, \cdots, \widehat{A}^{(i-1)}$
as shown in Lemma~\ref{lemma:w-i-of-k-nn}
have been obtained.
As a result, we have 
\begin{equation}
\begin{aligned}
C_{j}^{(i-2)} &= A_j^{(i-2)} A^{(i-3)} \cdots A^{(1)} ,  \,\,\,
C_j^{(i-1)} = A_j^{(i-1)} A^{(i-2)} \cdots A^{(1)} , \\
\widehat{C}_{j}^{(i-2)} &= \widehat{A}_j^{(i-2)} 
                                               \widehat{A}^{(i-3)} \cdots \widehat{A}^{(1)}
                                           = \frac{C_v^{(i-2)}}{
\left| \sum_{v = 1}^{d_{i-3}}{w_{j, v}^{(i-2)} C_{v, 1}^{(i-3)} }   \right|
} , \\
\widehat{C}_j^{(i-1)} &= \widehat{A}_j^{(i-1)} 
                                     \widehat{A}^{(i-2)} \cdots \widehat{A}^{(1)} 
                                     = \frac{C_v^{(i-1)}}{
\left| \sum_{v = 1}^{d_{i-2}}{w_{j, v}^{(i-1)} C_{v, 1}^{(i-2)} }   \right|
}. \\
\end{aligned}
\end{equation}


Now, the system of linear equations 
in Eq.~\eqref{eq:soe-of-k-nn} 
is transformed into
\begin{equation}	
\left\{
\begin{array}{c}
\sum_{u = 1}^{d_{i-1}}{\widehat{w}_{j, u}^{(i)} 
\frac{C_{u, 1}^{(i-1)}}{
\left| \sum_{v = 1}^{d_{i-2}}{w_{u, v}^{(i-1)} C_{v, 1}^{(i-2)} }   \right|
}
}
=
\frac{
\sum_{u = 1}^{d_{i-1}}{w_{j, u}^{(i)} C_{u, 1}^{(i-1)} }
}
{
\left|  
\sum_{u = 1}^{d_{i-1}}{w_{j, u}^{(i)} C_{u, 1}^{(i-1)} } 
\right|
} \\
\vdots   \\
\sum_{u = 1}^{d_{i-1}}{\widehat{w}_{j, u}^{(i)} 
\frac{C_{u, d_0}^{(i-1)}}{
\left| \sum_{v = 1}^{d_{i-2}}{w_{u, v}^{(i-1)} C_{v, 1}^{(i-2)} }   \right|
}
}
=
\frac{
\sum_{u = 1}^{d_{i-1}}{w_{j, u}^{(i)} C_{u, d_0}^{(i-1)} }
}
{
\left|  
\sum_{u = 1}^{d_{i-1}}{w_{j, u}^{(i)} C_{u, 1}^{(i-1)} } 
\right|
} \\
\end{array}
\right.
\end{equation}
When $d_0 \geqslant d_{i-1}$,
by solving this system, it is expected to obtain
\begin{equation}
\begin{aligned}
\widehat{A}_{j}^{(i)} 
&=
\left[ \widehat{w}_{j,1}^{(i)}, \cdots, \widehat{w}_{j,d_{i-1}}^{(i)} \right] \\
&= \left[
\frac{
w_{j,1}^{(i)} \times \left| \sum_{v = 1}^{d_{i-2}}{w_{1, v}^{(i-1)} C_{v, 1}^{(i-2)} } \right|
}
{
\left| \sum_{v = 1}^{d_{i-1}}{w_{j, v}^{(i)} C_{v, 1}^{(i-1)} }   \right|
},
\cdots,
\frac{
w_{j,d_{i-1}}^{(i)} \times 
\left| \sum_{v = 1}^{d_{i-2}}{w_{d_{i-1}, v}^{(i-1)} C_{v, 1}^{(i-2)} } \right|
}
{
\left| \sum_{v = 1}^{d_{i-1}}{w_{j, v}^{(i)} C_{v, 1}^{(i-1)} }   \right|
}
\right], \nonumber
\end{aligned}
\end{equation}
which is consistent with the expected value 
in Eq.~\eqref{eq:expected_Ai}.
\end{proof}


\section{Complexity of Hard-Label Model Extraction Attacks}
\label{appendix:attack_complexity}

For the $k$-deep neural network extraction attack,
its complexity is composed of two parts:
Oracle query complexity and computation complexity. 
Suppose that the number of neurons is $n = \sum_{i=1}^{k}{d_i}$.
Its input size, i.e., the size of $x$ is $d_0$.
And $k$ is the number of hidden layers.
The precision adopted by binary search is $\epsilon$
(refer to Section~\ref{subsec:find_DBP}).

\paragraph{\textup{\textbf{Oracle Query Complexity. }}}
For the $k$-deep neural network extraction attack,
we only query the Oracle in Steps 1 and 2 
(see Section~\ref{sec:overview}).

In Step 1,
if $c_n \times 2^n$ decision boundary points are collected,
then the number of queries to the Oracle is
$c_{\epsilon} \times c_n \times 2^n$,
where $c_{\epsilon}$ is a factor determined by
the precision $\epsilon$,
and $c_n$ is a small factor defined by the attacker.
In Step 2, 
for each decision boundary point $x$ collected in Step 1,
to recover the corresponding affine transformation 
(i.e., $\varGamma_{\mathcal{P}}$ and $B_{\mathcal{P}}$), 
we need to collect another $d_0 - 1$ decision boundary points.
Therefore, the times of querying the Oracle in this step is
$c_{\epsilon} \times c_n \times 2^n \times (d_0 - 1)$.
Based on the above analysis,
the Oracle query complexity of our $k$-deep neural network
extraction attack is $c_{\epsilon} \times c_n \times 2^n \times d_0$.
Note that $c_{\epsilon}$ 
is proportional to ${\rm log}_2^{ \frac{1}{\epsilon} }$.
Thus, the query complexity is about 
$\mathcal{O} \left( d_0 \times 2^n 
\times {\rm log}_2^{ \frac{1}{\epsilon} } \right)$.

\paragraph{\textup{\textbf{Computation Complexity.  }}}
For the $k$-deep neural network extraction attack,
when $n$ is large,
most computations are occupied 
by recovering neural network parameters, 
i.e., Steps 3 and 4 (see Section~\ref{sec:overview}).
Suppose that there are $N \leqslant c_n \times 2^n$ 
decision boundary points used to recover neural network parameters
after filtering duplicate affine transformations in Step 2.

In Step 3, 
to recover the weight vector of the $j$-th neuron 
in layer $i$ where $i \in \{2, \cdots, k+1\}$,
we need to solve a system of linear equations.
For convenience, let us ignore the difference in the sizes 
of the different systems of linear equations. 
Then, to recover all the weights $A^{(i)}$, 
a total of $n+1 - d_1 = \sum_{i=2}^{k+1}{d_i}$
systems of linear equations need to be solved.
In Step 4, to recover all the biases $b^{(1)}$,
only one system of linear equations needs to be solved.
Therefore, to obtain an extracted model,
we need to solve $n+2-d_1$ systems of linear equations.

There are two loops in the extraction attack.
First, we need to select $n+1$ out of $N$ 
decision boundary points each time.
More concretely,
to recover the weights $A^{(i)}$ in layer $i$,
we choose $d_i$ decision boundary points.
Then the number (denoted by $l_1$) of possible cases is
\begin{equation}
l_1 = 
\tbinom{N}{d_1} \times \tbinom{N - d_1}{d_2} \times \cdots \times
\tbinom{N - \sum_{i=1}^{k-1}{d_i}}{d_k} \times 
\tbinom{N - n}{1} 
\approx N^{n+1}, \,\, {\rm for} \,\, N \gg n . \nonumber
\end{equation}
Second, we need to guess $k$ signs when recovering all the weights,
i.e., there are $2^k$ cases.


Thus, the computation complexity
is about $\mathcal{O}\left( l_1 \times 2^k \times (n+2-d_1) \right)$.
When an appropriate precision $\epsilon$ 
(i.e., $\epsilon$ is small) is adopted,
we have $N \approx H < 2^n$, 
where $H$ is the number of possible model activation patterns
(refer to Lemma~\ref{lem:num_of_MAPs}).
Then, we further have
\begin{equation}
l_1 \times 2^k \times (n+2-d_1)
\approx N^{n+1} \times 2^k \times n
\approx n \times 2^{n(n+1) + k}.
\end{equation}
Thus, the computation complexity is about
$\mathcal{O} \left( n \times 2^{n^2+n+k} \right)$.



\section{Extraction on $1$-Deep Neural Networks}
\label{appendix:attack_on_1_deep_nn}

The parameters of the extracted $1$-deep neural network are as follows.

\begin{equation}
\begin{aligned}
\widehat{A}_i^{(1)} &= 
\left[ \widehat{w}_{i,1}^{(1)}, \cdots, \widehat{w}_{i,d_0}^{(1)} \right] =
\left[ \frac{w_{i,1}^{(1)}}{ \left|w_{i,1}^{(1)} \right|},
\cdots,
         \frac{w_{i,d_0}^{(1)}}{ \left|w_{i,1}^{(1)} \right|} \right],  \,\,
i \in \{1, \cdots, d_1\},		\\	
\widehat{b}^{(1)} &= [\widehat{b}_1^{(1)}, \cdots, \widehat{b}_{d_1}^{(1)}]
                                = \left[  
                                 \frac{b_1^{(1)}}{\left| w_{1,1}^{(1)} \right|},  \cdots,
			      \frac{b_{d_1}^{(1)}}{\left| w_{d_1,1}^{(1)} \right|}
                                \right], 		\\
\widehat{A}^{(2)} &= \left[ \widehat{w}_1^{(2)}, 
                                         \cdots, \widehat{w}_{d_1}^{(2)} \right]
                             = \left[ 
                                       \frac{w_1^{(2)} \left|w_{1,1}^{(1)}\right|}
                                              {\left| \sum_{i=1}^{d_1}{w_i^{(2)} w_{i,1}^{(1)}} \right|},
                                       \cdots,
                                       \frac{w_{d_1}^{(2)} \left|w_{d_1,1}^{(1)}\right|}
                                              {\left| \sum_{i=1}^{d_1}{w_i^{(2)} w_{i,1}^{(1)}} \right|}
                                \right],		\\	
\widehat{b}^{(2)} &= \frac{b^{(2)}}
                                        {\left| \sum_{i=1}^{d_1}{w_i^{(2)} w_{i,1}^{(1)}} \right|}  .\\
\end{aligned}
\end{equation}

Fig.~\ref{fig:attack_on_1_deep_nn} shows a diagram of
a victim model (2-2-1) and the extracted model. 

\begin{figure}[htb]
\centering
\includegraphics[width=1.0\textwidth]{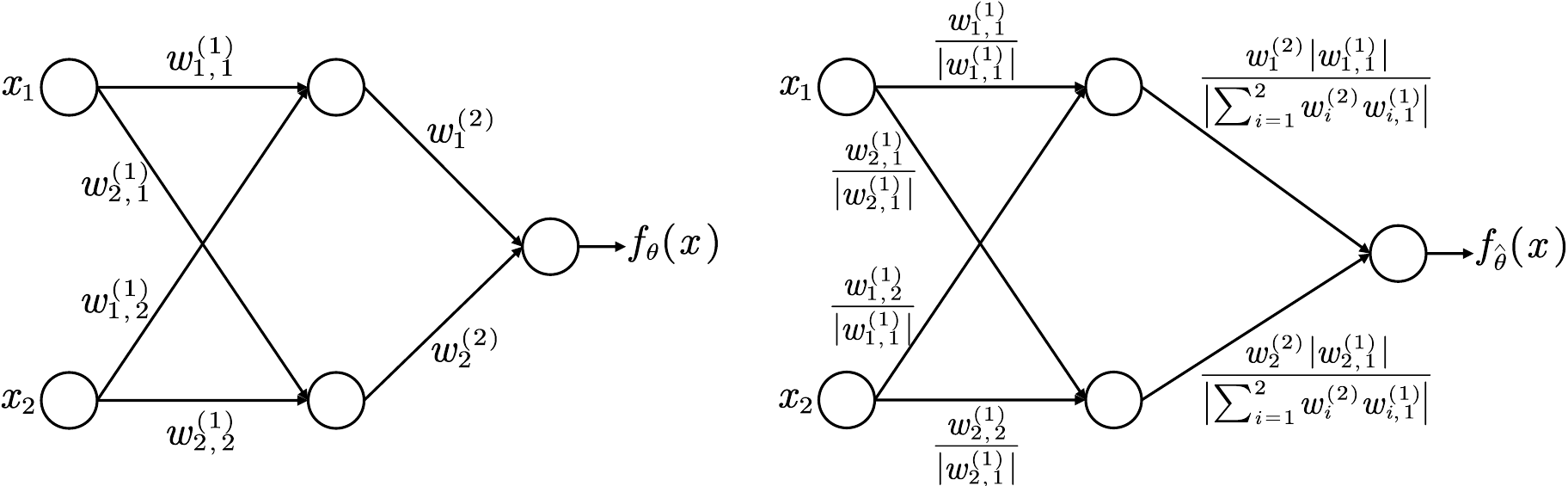}
\caption{
Left: the victim model $f_{\theta}$.
Right: the extracted model $f_{\widehat{\theta}}$.
}
\label{fig:attack_on_1_deep_nn}
\end{figure}


%
%
%
\bibliographystyle{splncs04}
\bibliography{reference}

\end{document}